\documentclass[10pt]{IEEEtran}
\usepackage{amsfonts,amsmath}
\usepackage{amssymb, amsmath, graphicx, epsfig, cite,subfigure,color,float}

 \newtheorem{theorem}{Theorem}[section]
 \newtheorem{corollary}[theorem]{Corollary}
 \newtheorem{lemma}[theorem]{Lemma}

 \newtheorem{definition}[theorem]{Definition}

\newcommand{\BlackBox}{\rule{1.5ex}{1.5ex}}  
\newenvironment{proof}{\par\noindent{\bf Proof\
}}{\hfill\BlackBox\\[2mm]} 

\definecolor{darkblue}{rgb}{0,0,0.7}

\title{Scalable and Anonymous Modeling of Large
Populations of Flexible Appliances}
\author{Mahnoosh Alizadeh, Anna Scaglione, Andy Applebaum, George Kesidis, and Karl Levitt \thanks{Supported by US DOE under CERTS load as a resource program, NSF SaTC CCF1229008, NSF 1228717. Email: malizadeh@ucdavis.edu. This work has been submitted in part to CDC 2014 as an invited presentation.}
}

\begin{document}
\maketitle

\begin{abstract}
To respond to volatility and congestion in the power grid, demand response (DR) mechanisms allow for shaping the load compared to a base load profile. When tapping on a large population of heterogeneous appliances as a DR resource, 
the challenge is in modeling the dimensions available for control. Such models need to strike the right balance between accuracy of the model and tractability.  
The goal of this paper is to provide a medium-grained stochastic hybrid model to represent a population of appliances that belong to two classes: deferrable or thermostatically controlled loads. We preserve quantized information regarding  individual load constraints, while discarding information about the identity of appliance owners. The advantages of our proposed population model are 1)  it allows us to model and control load in a scalable fashion, useful for ex-ante planning  by an aggregator or for real-time load control; 2) it allows for the preservation of the privacy of end-use customers that own submetered or directly controlled appliances.
\end{abstract}

\section{Introduction}\label{sec.intro}
Significant efficiency loss is experienced in the electricity market since demand is considered mostly inelastic.  
Demand response (DR) programs allow participants  to modify their load in response to economic incentives.  Dynamic pricing experiments seem to be promising \cite{hammerstrom2007pacific}, but the architecture for DR is not yet settled.

 Irrespective of implementations,   heterogeneous DR functionalities provided through the commercial and residential sectors are commonly subsumed in entities referred to as {\it aggregators}. 
Any aggregator with flexible demand will need to have day-ahead, hour-ahead and shorter term forecasts and risk models for the flexible loads it serves ({\bf ex-ante} models). It will also need a real-time {\it hybrid control} model to execute the optimal planned control in {\bf real-time}.  This need for a scalable model of load flexibility is an inherent property of any type of end-use demand management, e.g., direct load control/scheduling and dynamic retail pricing. However, when controlling a large population of heterogeneous appliances, developing such models can be a burden. Consequently, over-simplification of aggregate demand characteristics and constraints is common.
 
 The quest of this paper is to mathematically capture the ability of load profiles  resulting  from disparate, random uses to change into different shapes. We refer to the set of possible load shapes of an appliance or a population of appliances as {\it load plasticity}. Our idea is to build medium-grained models for the load plasticity of a population  based on quantizing the intrinsic controllability of a few archetypes. These accurate yet scalable load models can be beneficial in multiple stages.
As we will see in this paper, ex-ante, an aggregator can use stochastic load plasticities   to plan forward energy market purchases, make DR decisions, and quantify any associated risks. On the other hand, during real-time operation, they can be used to derive a compatible real-time control framework that is accurate and privacy-preserving.

\subsection{Contribution and prior art}
There has been extensive work on aggregate load modeling at the transmission level to forecast inelastic demand. However, as loads starts to become responsive at  large scales, the need for aggregate models of flexible demand becomes more significant. Without sufficiently accurate models, widespread DR can become a reliability hazard instead of a resource.

 Thus, the literature on DR includes a wide range of modeling options, corresponding to different resolutions in describing flexible load populations.  One approach is to preserve all details about the appliance state and constraints. The electric power consumption properties are either idealized (often as a battery) or described realistically via hybrid dynamical systems equations. Examples of adopters of such models for market analysis are  \cite{galus,6084772,6419868,6507354,6145671}. Detailed models are naturally advocated to shape the load in response to a price, using Home Energy Management Systems protocols \cite{mohsenian2010autonomous,kefayati2010efficient}, or to follow a desired load profile via centralized controllers or task scheduling rules that are optimal in some conditions; many papers considered Electric Vehicles (EV) \cite{kesidis,chen2012large,subramanian2012real,ramraja,6471273} and explored the optimality of simple policies such as  Least Laxity First (LLF) and Earliest Deadline First (EDF).

Another series of papers propose to capture the total flexibility of deferrable loads for planning and market interactions as a tank that needs to be fully charged by a certain deadline \cite{homer, 6426102,ortega}. All the specific characteristics of individual appliances except for their total energy consumption are discarded. Thus, the model has minimal computational cost  for ex-ante planning, but it is unsuitable for scheduling due to its coarseness and lack of information about individual appliances. 

These two groups of previous work sit at the two opposite sides of the spectrum in terms of accuracy and scalability. The goal of our work is to provide a medium-grained model that can be shifted across the spectrum as required for the application. The literature on flexible load modeling that most relates to our  paper dates back to  work on predictive modeling of the rebound peak due to emergency interruption of  a homogeneous population of Thermostatically Controlled Loads (TCL) in \cite{Chong}.  The challenge in population modeling is that different TCLs are at different states (temperatures) at a certain point of time and have different comfort requirements. The important idea to capture their load in a single mathematical model is to  classify these TCLs into
groups with similar states, resulting in a discrete formulation.
The concept of grouping similar TCLs  in \cite{Chong} was  later further refined by \cite{lu-chassin04,lu-chassin-windegren05, koch2011modeling}.  The temperature range of devices is quantized in bins, and the occupancy of each of these bins is tracked using a linear state space model. This approach leads to promising results in providing ancillary services. Recently \cite{hao2013generalized, globalSIP} used a similar model for heterogeneous TCLs with a control approach inspired by processor scheduling, and proposed LLF as a possible simple heuristic load control policy.   The authors in \cite{6407491} and our previous work \cite{smartgridcomm} discretized state of charge, rate and deadline constraints to map the problem of charge scheduling of heterogeneous electric vehicles  into a task scheduling problem.  Here we take this work a step further by providing a unified framework to describe a large population of {\it heterogeneous} loads that can be deferred, interrupted, and whose consumption rate can be controlled, as well as TCLs.

A population load model is the basis for modeling the response of an appliance population to DR signals and alter the load in the desired fashion. We will see that to build a population load model,  the aggregator must gather information about the energy usage time and pattern of different customers.  Obtaining this information through the Smart Meters, however, has raised significant concerns about customer privacy. Finding good compromises (see e.g. \cite{lalitha})  may be hard, as the need to ensure grid safety usually overwhelms arguments
  for confidentiality.  
 One of the byproducts of the scalable load model that we propose is that the data required to build the model is appliance-blind. We require no personal identification for monitoring and direct scheduling of appliances. Thereby, we can naturally enable customer anonymity. We will discuss this aspect of our model in detail for the interested reader. Previous work on privacy for the Smart Grid is focused nearly exclusively on the anonymization of meter readings, with a wide range of strategies intended to prevent household analysis having been proposed. This includes obscuring the readings themselves (e.g., \cite{5622047}, \cite{Acs2011IDD2042445.2042457}), or using multi-party aggregation techniques to hide the individual meters (e.g., \cite{Garcia2010PEV2050149.2050164,5622064,Kursawe2011PAS2032162.2032172}).  
Most closely related to our work is \cite{rottondi2013privacy}, in which the authors propose a privacy-friendly appliance scheduling protocol.


\textbf{Synopsis}: We introduce the basic elements of making scalable population models using the simplest type of appliances, i.e., ideal batteries, in Secion \ref{sec.loadmodel}. In Section \ref{pop.gen}, we generalize our models to realistic appliance categories. In Section \ref{sec.why}, we have a high-level discussion of the potentials of population load plasticities to help with DR planning and control. In Section \ref{sec.anoncom}, we discuss the communication requirements of submetering and direct scheduling frameworks that are based on hybrid population load models. Finally, in Section \ref{sec.numerical}, we demonstrate the usefulness of accurate yet scalable population models for an aggregator controlling EVs and TCLs.

\section{Modeling Aggregate Demand}\label{sec.modeling}\label{sec.loadmodel}

Most appliances are hybrid systems whose evolution is described by both continuous and discrete state variables \cite{alur1993hybrid}. 
One such continuous variable is the electric power that the appliance consumes, which we denote by ${\rm L}_i({\rm t})\in \mathbb{R}^{+}$ for an appliance indexed by $i$, and one such discrete variables is the power-switch state.  

To address the aggregate demand modeling problem systematically we introduce the concept of {\it load plasticity}. We show how the plasticity can be represented as an integer linear model in some important appliance categories.  
Then, building on these models for the individual appliances, we propose the idea of {\it clustering} appliances by quantizing the load plasticity into discrete classes. This leads to computationally scalable aggregate load plasticity models.

{\bf Notation}: We use roman font  for continuous variables $\rm x(t)$ and italic for discrete variables $x(t)$. Boldface
is used for vectors $\boldsymbol{x}$. We denote finite differences with respect to time as $\partial x(t)=x(t+1)-x(t)$. The unit step is denoted by $u(t)$, and   the Kronecker delta function is $\delta(t)$ (equal to 1 if the argument is equal to zero and equal to 0 otherwise). 
The symbol $\star$ denotes the discrete time convolution.  We use the notation ${\mathbb E}\{x\}$ to indicate
the expected value of a random variable $x$.

{\bf Vocabulary}: We say that an appliance {\it arrives} in the grid when it is first available to consume electricity, e.g., when an EV is parked and plugged in, or when a TCL  has to condition a space. The appliance may be immediately provided with electricity, as is the current service paradigm of the power grid. Alternatively, the service may be shifted to a later time under some type of demand management strategy, in which case we refer to the load as {\it deferrable}. The maximum amount of delay that an appliance can tolerate is captured by a deadline to finish its job, e.g., charge a battery in full. Alternatively, this maximum delay can be captured by a slack time. We define slack time as the amount of
time left between the end of an appliance's task and its deadline if the request is served immediately after arrival. For example, if a dryer cycle starts immediately after it is requested, the amount of time between the end of the cycle and its deadline represents the slack time.

\vspace{-0.3cm}
\subsection{Load Plasticity for the Ideal Battery}
We refer to the potential of ${\rm L}_i({\rm t})$ for being modified by control actions as {\it load plasticity}. Mathematically, if an appliance indexed by $i$ becomes available for load control at time ${\rm t}_i$, the load plasticity is a set ${\mathcal L}_i({\rm t})$ of load profiles that can be chosen by control actions taken at times ${\rm t>t_i}$. Any control is possible as long as the service quality remains within the preferences specified by the end-user.


\subsubsection{The Ideal Battery} The simplest type of load {\it plasticity} is that offered by an ideal battery that remains on indefinitely. An ideal battery indexed by $i$ first starts  to charge at a certain time $t_i$, it has an initial charge of ${\rm S}_i$ energy units, and a total charge capacity of ${\rm E}_i$ energy units. Denoting the  state of charge at time $t$ by $\rm x_i(t)$, the load plasticity is the set of load profiles:
  \begin{align}\label{lossless-battery}
  \rm
{\mathcal L}_i(\rm t)=\{\rm L_i(\rm t) |&  {\rm L}_i(\rm t) = \dot {\rm x}_i(\rm t), {\rm x}_i(t_i) = {\rm S}_i,\nonumber\\ & 0 \leq {\rm x}_i(t)\leq {\rm E}_i, t\geq t_i \}.
\end{align}

This model is analog and continuous and, thus, computationally infeasible. A natural step commonly taken in many control and communication applications that involve continuous values is to quantize signals, which is the approach we take in this paper. The reader will see that  quantizing continuous values and signals is the basic principle which allows us to provide a medium-grained scalable model for the load plasticity of a population of non-homogeneous appliances.

We start by quantizing time. If discrete time indices are separated by $\delta T$, we can use the index $t\in \mathbb{Z}$ such that ${\rm t}=t\delta T$.  
The second step consists of quantizing ${\rm E}_i$ and ${\rm S}_i$. We write the discrete counterparts of these variables as $(E_i, S_i) = (Q({\rm E}_i), Q({\rm S}_i))$, where $Q(.)$ represents a uniform quantizer with step of $\delta x$ energy units. For brevity of notation, we assume that both quantization steps are normalized to $\delta T=1$ and $\delta x=1$\footnote{Changing $\delta x$ and $\delta T$ to a generalized value is trivial and amounts to a scaling in all equations that map energy, e.g., $x(t)$, into power, e.g., $L(t)$. We avoid carrying on this scale as it unnecessarily burdens the notation. }.   Thus, the discrete version of the load plasticity for an ideal battery becomes:
  \begin{eqnarray}\nonumber
{\mathcal L}_i(t)&=&\{L_i(t) |  L_i(t) = \partial x_i(t), x_i(t_i) = S_i, \\
				& &x_i(t) \in \{0,1,\ldots, E_i\}, t\geq t_i \}\label{lossless-battery-discrete}.
\end{eqnarray}

Note that $ {\mathcal L}_i(t)$ can be fully specified by the following:
$$(t_i, S_i, E_i)$$
Hence,  {\bf ex-ante}, i.e., before the arrival of an appliance, the aggregator can model each load's plasticity, ${\mathcal L}_i(t_i)$, by characterizing the randomness of these three  parameters. The randomness of $t_i$ is captured by a random arrival process: 
\begin{equation}a_i(t)=u(t-t_i),\label{arrivaldef}\end{equation}
with $a_i(t)$ equal to 1 if appliance $i$ is present in the system. The randomness of $S_i$ and $E_i$ is captured by a probability law $p(S,E; t_i)$ that reflects historical statistical information about requests.
  For a {\bf real-time} controller, the state and thus the future load plasticity of an appliance  changes  as control actions are taken, and the residual set  ${\mathcal L}_i(t)$ at  time $ t>t_i$ is given by replacing  $S_i$ with  $x_i(t)$ and $t_i$ with $t$.

\subsubsection{Population model for the ideal battery}\label{idealbat.pop}
The difference between modeling a single ideal battery and a population of ideal batteries is in capturing the effects of non-homogeneous $(t_i, S_i, E_i)$ in the description of the aggregate population load plasticity in a scalable fashion. To achieve this goal, we will try to group appliances that share similar load plasticities together, allowing us to discard information about individual appliances.

To start, suppose that we are modeling a population of batteries indexed by $i \in {\cal P}_E$ that all have the same capacity $E$, but have non-homogeneous $S_i$ and $t_i$. 
Using \eqref{lossless-battery-discrete} and \eqref{arrivaldef}, we can write the total load of the batteries as:
\begin{equation}\label{pop}
L(t)=\sum_{i\in {\cal P}_E}L_i(t)=\sum_{i\in {\cal P}_E}\partial x_i(t) a_i(t).
\end{equation}

To avoid tracking individual appliances and lower the complexity of \eqref{pop}, the basic idea first used  in \cite{chong85} is
to track the number of appliances $n_x(t)$ that are present in the system at time $t$ and are at quantized state $x$, discarding information that identifies individual appliances. This {\it aggregation} is what enables our anonymous modeling and control techniques. Mathematically, we can write $n_x(t)$ as
\begin{equation}\label{numberineachstate}
 n_x(t)=\sum_{i\in {\cal P}_E}\delta(x_i(t)-x)a_i(t),~x=1,\ldots,E
\end{equation} 
Similarly, we denote the total number of batteries that arrive in the system with an initial state of charge equal to $x$ at or before time $t$ as $a_x(t)$. The value of $a_x(t)$ should be tracked in real-time as an input to the population model and can be written in terms of individual arrival processes $a_i(t)$ as:
\begin{equation}\label{apx}
 a_x(t)=\sum_{i\in {\cal P}_E}\delta(S_i-x)  a_i(t).
\end{equation}
We refer to $a_x(t)$ as the {\it arrival process for state $x$}.

Next, we directly tie the evolution of $n_x(t)$ and $a_x(t)$ to the total load $L(t)$, removing all dependence on $x_i(t)$ and $a_i(t)$.

\begin{lemma}\label{nxlemma}{\it 
The following relationship holds between $n_x(t)$ and the load $L(t)$:
\begin{align}\label{occur}
L(t) & = \sum_{x=0}^{E}\left[\left(\sum_{x'=x}^{E} \partial n_{x'}(t) \right) - (x+1)\partial a_{x}(t)\right].
\end{align}
}
\end{lemma}
 See Appendix \ref{prooflemma} for proof.

The movement of appliances from one state to another is what determines $\partial n_x(t)$. We capture this next.

\begin{definition}{\it  We denote by $d_{x,x'}(t)$ the number of batteries that go from state $x$ to state $x'$ at time $t$. We we refer to this number as the {\it switch process from state $x$ to $x'$}. We define $d_{x,x}(t) = 0, ~\forall t,x$, and $d_{x,x'}(0) = 0,~\forall x,x'$.}
\end{definition}
\begin{corollary}{\it
The occupancy $n_x(t)$ and aggregate load $L(t)$ in terms of $d_{x,x'}(t)$ are: }
\begin{align}\label{nxt2}
n_{x}(t+1)&=  a_x(t+1) + \sum_{x' = 0}^E [d_{x',x}(t)- d_{x,x'}(t)]
\\\label{occur3}
L(t) &= \sum_{x=0}^E \sum_{x'=0}^E (x'-x)\partial d_{x,x'}(t)
\end{align}
\end{corollary}
\begin{proof}
The occupancy at time $t+1$ should include the previous occupancy plus new arrivals from other states or from outside, minus the population that exits the state:
\begin{align}\label{nxt}
\partial n_{x}(t)&=\partial a_x(t) + \sum_{x' = 0}^E \partial [d_{x',x}(t)-d_{x,x'}(t)]
\end{align}
which leads to \eqref{nxt2} if summed over time. If we substitute the value of $\partial n_x(t)$ in \eqref{nxt} into \eqref{occur}, we get \eqref{occur3}.\end{proof}
Thus, the load plasticity of a $L(t)$ of a population of ideal batteries can be presented in terms of the $d_{x,x'}(t)$'s, under appropriate constraints:
  \begin{eqnarray}\nonumber
{\mathcal L}_E(t) = \left\{L(t) |  L(t) = \sum_{x=0}^E \sum_{x'=0}^E (x'-x)\partial d_{x,x'}(t), \right. &&\\ \left.\partial d_{x,x'}(t) \in \mathbb{Z}^+, \sum_{x' = 1}^{E}\partial d_{x,x'}(t) \leq n_x(t)\right\},&&\label{occur4}
\end{eqnarray}
where $n_x(t)$ is given by \eqref{nxt2}. The second constraint ensures that the number of appliances that leave  state $x$ at time $t$ is less than or equal to the number of appliances present in state $x$ at time $t$, i.e., $n_x(t)$.

%

One advantage of aggregating constraints for appliance populations into one model is that 
integrality constraints on $d_{x,x'}(t)$'s can be safely relaxed at large scales. This makes the population model linear and less computationally expensive.

Now, we go back to address non-homogeneity of the battery capacity $E_i$ across the population. Generally, the parameters that describe an individual appliance load can be divided into two groups: one set of parameters, denoted by $\boldsymbol{\kappa}_i$, describe the initial state of  control variables, e.g., the state of charge of a battery or the representative temperature of a TCL. Changing these quantities affects the load plasticity of an appliance only in a transient fashion. Appliances that only differ in terms of these initial parameters can be bundled together in a single population model for the load plasticity, like the batteries discussed above, for which $\boldsymbol{\kappa}_i = (t_i,S_i)$. Another set of parameters, denoted by $\boldsymbol{\theta}_i$, define the constraints that change the underlying structure of load plasticity. These parameters can include the physical constraints of the device, e.g., capacity of a battery or the wattage of a TCL, or the quality of service required by the user, e.g., the need to fully charge the battery by 8AM, or the need to keep the temperature inside a certain comfort band. Next, we address how we handle this change in nature of load plasticity by grouping loads in {\it clusters}.

\subsection{Load Clusters for Modeling Non-homogenous Populations}\label{cluster.sec}
%

The  proposed  hybrid scalable load model  can be naturally generalized to handle a population of appliances with non-homogeneous $\boldsymbol{\theta}_i$  by quantizing $\boldsymbol{\theta}_i$. Thereby, we assume that appliances convey the parameters that capture their inherent constraints and the users' consumption preferences from a finite number of choices. To provide an example, this means that we can quantize the battery capacity to levels 1kWh apart. Thus, our model will not distinguish between two batteries with capacities 5.4 and 5.25 kWhs.

 We bundle requests with similar constraints in {\it clusters} indexed by $q=1, \ldots, Q$:
\begin{equation}
\boldsymbol{\theta}_i \xrightarrow[\mathcal{Q}]{\mbox{Quantize}} \boldsymbol{\theta}_q \xrightarrow[\mathcal{I}] {\mbox{Cluster index}} q.
\end{equation} 
The level of quantization error can be controlled by modifying $Q$ or $\boldsymbol{\theta}_q$'s, and is the knob that controls the complexity and accuracy of the aggregate model. This amounts to quantizing both the feature space and constraints of each appliance. 

For the ideal battery with no charging deadline, $\boldsymbol{\theta}_i = (E_i)$. This means that two batteries that share the same quantized capacity ($E^q$) and set of possible states ${\mathcal X}^q=\{0,1,\ldots,E^q\}$, but may vary in $\boldsymbol{\kappa}_i = (t_i,S_i)$,  belong to the same cluster $q$.  However, if two batteries have different capacities, they are not bundled together in the same cluster, since their load plasticities are different in nature. This is due to the fact that ${\mathcal X}^q$ differs from one cluster to the next.

We use a superscript $q$ to refer to any previously defined quantity for cluster $q$. Thus, generalizing \eqref{occur4}, the load plasticity of a non-homogeneous population of ideal batteries is,
  \begin{eqnarray} \label{occur5} \nonumber
{\mathcal L}(t) = \left\{L(t) |  L(t) = \sum_{q=1}^Q\sum_{x=0}^{E^q} \sum_{x'=0}^{E^q} (x'-x)\partial d^q_{x,x'}(t)\right.\nonumber 
& & \\ \left. \partial d_{x,x'}^q(t) \in \mathbb{Z}^+, \sum_{x' = 1}^{E^q}\partial d_{x,x'}^q(t)\leq n^q_x(t) \right\}& &  \nonumber
\end{eqnarray}
with
\begin{align}
n_{x}^q(t)&=  a^q_x(t) + \sum_{x' = 0}^{E^q} [d^q_{x',x}(t-1)- d^q_{x,x'}(t-1)]
\end{align}

This gives us a hybrid stochastic model for the load plasticity of a population of ideal batteries with no charge deadline. In Section \ref{pop.gen}, we use the clustering concept introduced here to provide population models for appliances with different load shapes and user preferences, e.g., charging deadlines.

\vspace{-0.4cm}

\subsection{Stochastic Decision Models}
The objective of this paper is not to propose a specific approach for demand management, e.g., \cite{mohsenian2010autonomous,kefayati2010efficient,kesidis,chen2012large,subramanian2012real,ramraja,6471273}. Rather, we offer scalable and  stochastic aggregate load models that can be used in conjunction with different DR schemes. The DR designer can use load plasticities with any preferred multistage stochastic optimization framework or be myopic. Obviously the computational complexity varies among different frameworks. Notice that the arrivals of appliances are naturally going to be non-stationary random processes. Thus, if the population is large, Monte Carlo sampling and Sample Average Approximation (SAA) on the random arrivals \cite{kleywegt2002sample} are  likely choices to reduce the stochastic optimization size.

Note that one would need to capture the statistics of the random arrival processes $a^q_x(t)$ in order to model stochastic load.  As an example of how one can gather and forecast these arrival statistics, the reader can see e.g. \cite{alizadeh2013ev}, which looks at forecasting and modeling Electric Vehicle (EV) charging requests arrivals using non-homogeneous Poisson processes.

\vspace{-0.1cm}

\section{Population Models for Realistic Loads}\label{pop.gen}
Real appliance loads are more constrained in terms of controllability than the ideal battery. Some loads are merely {\it deferrable} but uninterruptible. Many loads can be only turned ON and OFF, i.e., $L_i(t)\in\{0,G_i\}$. Some loads resemble the ideal battery, e.g., electric vehicles, but often include limitations on the rate of charge (change of state) and deadlines.

Here we build on the ideal battery results to expand our modeling effort to more realistic appliances. We bundle appliances in each category $v$ together in a single population model, and denote the load plasticity of each population by $\mathcal{L}^v(t)$. There are $Q^v$ clusters for each category $v$.

In general, an aggregator can serve $V$ different categories of loads. Given the population load plasticity of each category ${\mathcal L}^v(t)$, it is straightforward to write the load plasticity of the total demand served by an aggregator as the set:
\begin{align}\label{totplas}
{\mathcal L}(t)=\left\{L(t)|  L(t) = L^{I}(t) + \sum_{v=1}^V L^v(t), L^v(t) \in {\mathcal L}^v(t) \right\}
\end{align}
where $L^{I}(t)$ denotes the inflexible demand served by the aggregator at time $t$.

\subsection{Rate-Constrained Instantaneous Consumption (RIC)} \label{cat1}
The first category looks at a general battery with a maximum charging rate and a deadline to be charged by a certain amount. This could capture electric vehicle charging requests, including Vehicle to the Grid (V2G) applications. We consider a non-homogeneous population of batteries, each characterized by the vector $(t_i, X_i , E_i, \chi_i, \rho_i, G_i)$. The new parameters over those of the ideal battery are as follows: $\chi_i$   denotes the deadline for battery $i$ to receive at least $\rho_i$ percent of full charge, and $G_i$ denotes the maximum rate at which battery $i$ can be charged/discharged. The load plasticity of battery $i$ is:
  \begin{eqnarray}\nonumber
{\mathcal L}_i(t_i)&=&\{L_i(t) |  L_i(t) = \partial x_i(t), x_i(t_i) = S_i, \\
				& &x_i(t) \in \{0,1,\ldots, E_i\}, x_i(\chi_i) \geq \rho_i E_i, \nonumber \\
& & -G_i \leq \partial x_i(t) \leq G_i, t\geq t_i \}  \label{lossless-battery-discrete2}
\end{eqnarray}
By setting $\chi_i = \infty$, $\rho_i = 1$, and $G_i = E_i + 1$ the model is equivalent to the ideal battery.
Mapping this type of load plasticity into a population model is straightforward following the discussion in Sections \ref{idealbat.pop} and \ref{cluster.sec}. We cluster batteries using $\boldsymbol{\theta}_i = (E_i, \chi_i, \rho_i, G_i)$, and denote the parameters associated with cluster $q$ as $(E^q, \chi^q, \rho^q, G^q)$. Then, we need to modify the population model to handle the additional two constraints in \eqref{lossless-battery-discrete2} over that of the ideal battery in \eqref{lossless-battery-discrete}.

In each cluster $q$:
\begin{enumerate}
\item Batteries cannot move from state $x$ to any state $x'$ in one time step. Due to rate limitations, at each time step, a battery can move by a maximum of $G^q$ states, i.e. it can draw or deliver power at a rate $|\partial x(t)|\leq G^q$. Thus, the switch processes from and to state $x$ are only defined for $x' \in {\mathcal S}_x$ with 
\begin{equation}
{\mathcal S}_x =\{x - G^q, x - G^q+1, \ldots, x + G^q\};
\end{equation}
\item All appliances in cluster $q$ should be in a state $x \geq \rho^q E^q$ by time $\chi^q$. This translates into $n_x^q(\chi^q) = 0, \forall x < \rho^q E^q$.
\end{enumerate}

Consequently, the load plasticity of the population is:
  \begin{eqnarray} \label{occur5} 
{\mathcal L}^v(t) = \Big\{L(t) |  L(t) = \sum_{q=1}^{Q^v}\sum_{x=0}^{E^q} \sum_{x'\in \mathcal S_x}(x'-x)\partial d^q_{x,x'}(t),  & & \nonumber\\ \partial d_{x,x'}^q(t) \in \mathbb{Z}^+,\!\!\sum_{x'\in \mathcal S_x}\partial d_{x,x'}^q(t)\leq n^q_x(t),\nonumber & & \\  n_x^q(\chi^q) = 0, \forall x < \rho^q E^q\Big\},~~~~~~~~~~~~~~~& & 
\end{eqnarray}
with
\begin{align}\label{nx}
n_{x}^q(t)&=  a^q_x(t) + \sum_{x'\in \mathcal S_x}[d^q_{x',x}(t-1)- d^q_{x,x'}(t-1)].
\end{align}

\subsection{Interruptible service (IS)}  \label{cat2}
This category is more constrainted since it only allows batteries to charge at a constant rate $G_i$. The charge can however be interrupted multiple times.
 This category best models pool pumps or EVs that can only be charged at certain charging levels, e.g., 1.1 kW or 3.3 kW. Due to the similar nature of this category with the RIC, we simply present the population model. We cluster loads based on $\boldsymbol{\theta}_i = (E_i, \chi_i, \rho_i, G_i)$, and the population plasticity is:
 \begin{align} %
{\mathcal L}^v(t) \!=\! &\Big\{L(t) |  L(t) \!= \!\sum_{q=1}^{Q^v}\sum_{x=0}^{E^q} \!(x'\!-\!x)\partial d^q_{x,x'}(t)|_{x'=\min\{x+G^q,E^q\}}, \nonumber  \\ 
&\partial d_{x,x'}^q(t) \in \mathbb{Z}^+, \partial d_{x,\min\{x+G^q,E^q\}}^q(t)\leq n^q_x(t),\nonumber \\ 
&  n_x^q(\chi^q) = 0, \forall x < \rho^q E^q\Big\} \label{occur6} 
\end{align}

\subsection{Thermostatically controlled loads (TCLs)}\label{cat3}
 TCLs keep a representative temperature $x_i(t)$ within the boundaries of a comfort or safety band $[x^*_i-B_i/2,x_i^*+B_i/2]$ in the time window $[\chi^{s}_i,\chi_i^{e})$ of the day. We denote by $x^*_i$ the center and by $B_i$ the width of the comfort band. The information on $[\chi^{s}_i,\chi_i^{e})$ can be used for preconditioning (precooling and preheating). Thus, the time frame at which the thermostat operates the appliance can be larger than  $[\chi^{s}_i,\chi_i^{e})$. We assume that a TCL can be operated within a time frame $[t^{s}_i,t_i^{e})$. The availability of a TCL for control can be modeled through defining an arrival and departure process for TCL $i$:
\begin{align}
a_i(t) = u(t- t^{s}_i),~~r_i(t) = u(t- t^{e}_i).
\end{align}
TCLs are controlled with a cycling switch that turns them off/on, i.e. $L_i(t) \in \{0,G_i\}$.   Emulating the approach proposed in \cite{chong85}, we take a   quantized version of $x_i(t)$ to capture the state of TCL $i$ at time $t$. 
Using the model proposed in \cite{ihara1981physically}, we capture any random effects through a noise term denoted by $\alpha_i(t)$\footnote{The additive noise term in the original model proposed in \cite{ihara1981physically} is subsumed in the random process $\alpha_i(t)$.}. Thus, for unit $i$ we have:
\begin{align}
{\mathcal L}_i(t)=\Big\{ L_i(t)&|\partial x_i(t)=-k_i x_i(t) +\alpha_i(t)+ b_i(t) G_i,\label{tcls}\\
					& b_i(t)\in \{0,1\}, L_i(t)=b_i(t)G_i, \forall t \in [t_i^s,t_i^e) \nonumber\\
					&|x_i(t)-x^*_i|\leq B_i/2,~  \forall [t]_{24H} \in [\chi_i^s,\chi_i^e)  \Big\}\nonumber
\end{align}
where $b_i(t) = \{0,1\}$ denotes the off/on status of TCL $i$. The rate of heat gain is taken to be $G_i$ Btu/h.
The noise expected value is ${\mathbb E}[\alpha_i(t)]= x_{{\mathrm amb}}(t) k_i$ where $ x_{{\mathrm amb}}(t)$ is the ambient
temperature. We later assume that $x_{{\mathrm amb}}(t)$ varies slowly in time, and can be approximated with a constant in each hour. 

TCL load control, mainly through direct setpoint changes, has been widely studied, e.g., \cite{6149125}. Aggregate TCL models are fairly complex compared to the models introduced thus far in this work, because: 1) state transitions are {\it random} due to the interaction with the randomly changing ambient temperature and customer activities; 2)
the control action can only indirectly affect the number of appliances that migrate from one state to the other.   
In fact, let  $p_{\alpha_i}(\alpha;t)$ be the PMF of the quantized random process $\alpha_i(t)$ at time $t$. 
According to \eqref{tcls}, a transition from state $x$ to state $x'$ occurs at time $t$ if and only if
\begin{align}\alpha_i(t)&=x' - (1- k_i) x  - b_i(t) G_i\end{align}
The impact of switching the control $b_i(t)\in\{0,1\}$ is in changing the PMF that dictates how appliances move from one state $x$ to $x'$, denoted by $P_i(x'|x;t;b_i)$:
\begin{align}
P_i(x'|x;t;b_i(t))&=p_{\alpha_i}\left(x'-x(1-k_i)-b_i(t) G_i\right).
\end{align}
This models the probability that TCL $i$ goes from state $x$ to $x'$ at time $t$ as the sum of two mutually exclusive random events:
\begin{enumerate}
\item If the appliance is on, i.e., $b=1$, it can go from state $x$ to state $x'$ with success probability $P_i(x'|x;t;1)$;
\item If the appliance is off, i.e., $b=0$, it can go from state $x$ to state $x'$ with success probability $P_i(x'|x;t;0)$;
\end{enumerate}

To generalize this individual load plasticity to a population model, let us introduce:
\begin{definition}{\it Let
$D_{x,x'}(t)$ denote the random number of TCLs that switch from state $x$ to state $x'$ in the interval of time $t$. 
Let $n_{x,b}(t)$ indicate the number of TCLs in state $x$ that have a status of $b=\{0,1\}$ at time $t$.} 
\end{definition}

Note that $n_{x,b}(t)$ is exactly analogous to the model in \cite{mathieu} and to the priority stack in \cite{hao2013generalized}. What our model incorporates is the concept of clustering characteristics, which other authors commonly replaced with simplifying mean field assumptions. 
The discussion above implies that $D_{x,x'}(t)$ is the total number of TCLs that go from state $x$ to $x'$ under one of two Bernoulli random trials with success probability $P_i(x'|x;t;b_i(t))$. In general, the switching probabilities $P_i(x'|x;t;b_i(t))$ are different for heterogeneous TCLs.  
Thus, following our previous quantization approach, we cluster the TCLs based on their basic characteristic tuple $\boldsymbol{\theta}_q = (G^q, x^{*q}, B^q,\chi^{s,q},\chi^{e,q})$,  but also based on the statistics $p_{\alpha^q}(\alpha;t)$. Consequently, we can define per-cluster switching PMFs $P^q(x'|x;t;b)$, which gives the probability of TCLs with switch position $b$ in cluster $q$ going from state $x$ to $x'$.  Consequently, $D_{x,x'}(t)$ is the sum of two random components:
\begin{enumerate}
\item The number of TCLs among $n_{x,1}(t)$ TCLs in the ON position that go from state $x$ to state $x'$ with probability $P^q(x'|x;t;1)$;
\item The number of TCLs among $n_{x,0}(t)$ TCLs in the OFF position that go from state $x$ to state $x'$ with probability $P^q(x'|x;t;0)$;
\end{enumerate}
Thus, $D_{x,x'}^q(t)$ at time $t$ is the sum of two Binomial random variables ${\mathcal B}(n_{x,b}^q(t),P^q(x'|x;t;b))$, with $b\in \{0,1\}$.  Hence,
\begin{align}
{\mathbb E}\{D_{x,x'}^q(t)|n^q_x(t)\}&=\sum_{b=0}^1n_{x,b}^q(t)P^q(x'|x;t;b)
\end{align}
Observe that the control action available here is picking the number of appliances that are turned on/off in each state $x$, i..e, $n_{x,0}^q(t)$ and $n_{x,1}^q(t)$, subject to:
\begin{equation}
n_{x,0}^q(t) + n_{x,1}^q(t) = n_{x}^q(t).
\end{equation}
If we denote the set of  all temperature bins in cluster $q$ by ${\mathcal S}^q$, the total number of appliances in state $x$ in cluster $q$, $n_{x}^q(t)$, is governed by the dynamics: 
\begin{align}
n^q_x(t)&\!=\!a_x^q(t)\!-\! r_x^q(t)+\!\!\!\!\sum_{x'\in {\mathcal S}^q} \!\!\!D_{x',x}^q(t-1)\!-\!D_{x,x'}^q(t-1).\label{nqxdyn}
\end{align}
where $a_x^q(t)$ and $r_x^q(t)$ are the arrival and departure processes for state $x$ of cluster $q$. They respectively count the number of TCLs that first become available and unavailable for control before or at time $t$, when they are at state $x$ of cluster $q$.

The comfort band constraint translates into
\begin{equation}
\forall |x-x^{*q}|>B^q/2~~\rightarrow~~{\rm Pr}(n_x(t)=0)\geq \eta, 
\end{equation}
where $ \eta$ is close to one, indicating that violations are rare if not outright impossible. 

Consequently, the stochastic aggregate plasticity of the population of heterogeneous TCLs is:
  \begin{eqnarray} \label{occur5} 
{\mathcal L}^v(t) = \Big\{L(t) |  L(t) = \sum_{q=1}^{Q^v} G^q \sum_{x\in {\mathcal S}^q}n_{x,1}^q(t), n_{x,1}^q(t) \in \mathbb{Z}^+&&
\\~n_{x,0}^q(t) + n_{x,1}^q(t) = n_{x}^q(t)~, 
n^q_x(t)=\mbox{see \eqref{nqxdyn}}&&\nonumber\\
{\mathbb E}\{D_{x,x'}^q(t)|n^q_x(t)\}=\sum_{b=0}^1n_{x,b}^q(t)P^q(x'|x;t;b); &&\nonumber\\
\forall x: |x-x^{*q}|>B^q/2,~~\forall [t]_{24H}\in [\chi^{s,q},\chi^{e,q})&&\nonumber\\~~~~~~~~~~~~~~~~~~~~~~~~~\rightarrow~{\rm Pr}(n^q_x(t)=0)\geq \eta \Big\}  \nonumber&&
\end{eqnarray}

\subsubsection*{Real-time Coarse Clustering} Contrary to all other population models discussed in the paper, \eqref{occur5}  is rather complex, mainly due to the fact that energy storage in a TCL is lossy. Note that \eqref{occur5} should naturally be paired with a model-predictive control strategy since its essence is in modeling the future random changes in temperature $x$  due to external noise and heat loss.
The noise term $\alpha^q(t)$ in \eqref{tcls} can be replaced with its expectation $x_{{\mathrm amb}}(t)a^q$ when using the model ex-ante, and this could 
considerably simplify planning decisions by lowering the number of clusters $Q^v$.  However, for real-time direct scheduling, using \eqref{occur5} can be cumbersome.

Thus, we propose a coarser real-time clustering method that is considerably less computational and communication intense compared to \eqref{occur5}.
Suppose that we limit ourselves to myopic policies, which is common practice in the field of routing and scheduling for large number of tasks. Note that the ultimate goal of a real-time direct scheduler would be to plan the switching events of each TCL under the constraint that the temperature does not go outside of the comfort band. Thus, one could envision that each TCL switching event can be scheduled by a deadline that can be predicted locally using \eqref{tcls}. 
In this case, a good myopic on-line scheduling technique can be derived by assigning an index that describes the imminence of each switching deadline communicated to the scheduler by the TCLs.  Then, the scheduler can cluster the population based on the quantized deadline index, rather than all of the parameters
in \eqref{occur5}.  

Thus, in the real-time coarse clustering approach, TCLs are clustered based on the current state of their switch (ON or OFF) and based on how imminent is their next switching event. Assume that $\alpha_i(t)$ is known to the $i$th customer and it varies slowly. Note that the two boundaries that force the appliance to switch state can be written as a function of the switch state $b_i(t)$ in equation \eqref{tcls}:
$$x_i^*-(-1)^{b_i(t)}\frac{B_i}{2}.$$
Then, solving the equation $x_i(\tau_{i}+t)=x_i^*-(-1)^{b_i(t)}\frac{B_i}{2}$ using the linear dynamics in \eqref{tcls}, the time at which appliance $i$ will reach
its upper or lower boundary is:
\begin{equation}
\tau_{i}(t)=\frac 1 {k_i}\ln\left(
\frac{x_i(t)-b_i(t)\frac{G_i}{k_i}-\frac{\alpha_i(t)}{k_i}}
{x_i^*-(-1)^{b_i(t)}\frac{B_i}{2}-b_i(t)\frac{G_i}{k_i}-\frac{\alpha_i(t)}{k_i}}
\right).
\end{equation}
Given that each TCL can locally calculate $\tau_{i}(t)$, we propose to use as the state of the TCL the pair
$(\tau_{i}(t),b_i(t))$ that captures the switching deadline for each TCL, irrespective of all other cluster parameters. Note that the power consumption $G$ is also ignored in this abstraction. Thus, the scheduler will have to use the average power consumption $\bar{G}$ of all TCLs in the population as an approximation of how much the load will be affected by each turn on/off event.
Just as before, $\tau_i$ is quantized and replaced with the closest  point $Q(\tau_i)$ 
in the set $\{0,\delta \tau,\ldots,(N-1)\delta \tau\}$.
Therefore, if we now use $\tau$ as a discrete index, the population model is:
\begin{equation}\label{coarsecluster}
n_{\tau,b}(t)=\sum_{i\in {\mathcal P}} \delta(\tau-Q(\tau_i(t)))\delta(b-b_i(t))(a_i(t) - r_i(t)).
\end{equation}
We will see an example of how these two models can be used in Sections \ref{sec.TCLex-ante} and \ref{sec.TCLonline} respectively.

\subsection{ Non-interruptible Deferrable Service (NID)}\label{cat4}
For this category of loads we assume that, once on, the appliance needs to complete a cycle and will automatically follow a preset load profile. However, the starting time can be shifted across hours within customer-specified limits. This best models appliances such as washer/dryers, and non-interruptible EV charging. The load is the output of a hybrid system captured by two hybrid states, one being the system at rest, and the other is the system evolution once ON, captured by a continuous consumption profile.

Each appliance $i$ in this category is characterized by $(t_i,\chi_i,\ell_i(t))$, where $t_i$ is the arrival time, $\chi_i$ is the maximum tolerable delay to start consumption, and $\ell_i(t)$ is a pulse that captures the load profile of appliance $i$ (if it is turned on at $t=0$).
Hence, the only control available is shifting the load to start at time $t_i \leq \tau_i \leq t_i+\chi_i$. Let $\ell_i(t)$ be the load signal if it starts at time 0. The load
plasticity is simply \begin{align}
{\mathcal L}_i(t)=\{L_i(t)| L_i(t)=\ell_i(t-\tau_i),t_i \leq \tau_i \leq t + \chi_i\}
\end{align}
The description above can be replaced with following integer linear model, based on the state $x_i(t)$ of the ON switch:
\begin{align}
{\mathcal L}_i(t)=&\{L_i(t)| L_i(t)=\ell_i(t)\star \partial x_i(t), x_i(t) \in \{0,1\},\label{eq.defer}\\
	& ~x_i(t)\geq a_i(t-\chi_i), ~x_i(t-1)\leq x_i(t)\leq a_i(t) \}.\nonumber
\end{align}
where $x_i(t) \in \{0,1\}$ denotes the state of the appliance, with 0 and 1 respectively corresponding to off and on status, and $\star$ denotes the convolution operation.  
Note that $x_i(t)$ can only be in the form of a step function, i.e., once $x_i(t) = 1$ at time $\tau_i$, if has to remain 1 for the remainder of time due to the constraint $x_i(t-1)\leq x_i(t)$.   
The convolution with $\partial x_i(t)=\delta(t-\tau_i)$ yields
$L_i(t)=\ell_i(t-\tau_i).$ The constraints ensure that an appliance can only be turned on after it arrives, i.e., when $a_i(t) = 1$, and that it cannot be turned off once on, i.e., appliances do not move from state $x=1$ to state $x=0$.


To build a population model, we cluster appliances based on quantized versions of the load profiles $\ell^{q}(t)$, plus the maximum tolerable delay $\chi^q$ they can tolerate. The total number of arrivals and activations of cluster $q$ up to time $t$ are denoted by $a^q(t)$ and $d^q(t)$. Similar to the individual model, the load of cluster $q$ is the convolution of the number of activations with the pulse $\ell^q(t)$.
The population plasticity is the sum of contributions from each cluster $q$:
\begin{align}\label{lint}
{\mathcal L}^v(t) \!=\! \Big\{&L(t) |  L(t) \!= \!\sum_{q=1}^{Q^v} \ell^q(t)\star \partial d^q(t), d^q(t)\in\mathbb{Z}^+ \\ 
& ~d^q(t)\geq a^q(t-\chi^q), ~d^q(t-1)\leq d^q(t)\leq a^q(t)
\Big\} \nonumber
\end{align}


\subsection{General Hybrid System (GHS)}\label{cat5}
The modeling for the NID can be extended to include
more hybrid states $x\in {\mathcal X}^q$ than just a single switch. We can extend the associated dynamics  of the system as it enters in a certain hybrid state to state specific load injection profiles $\ell_x^q(t)$. This is useful for modeling non-interruptible tasks that follow each other with a single shared deadline. A simple example for residential use is the washer and dryer cycle. 

In the general case, there will be a directed graph ${\mathcal G}^q=({\mathcal S}^q,\mathcal{E}^q)$ with discrete states $x\in {\mathcal S}^q$ and edges $\mathcal{E}^q$ describing the hybrid automaton and constraining the switching $\partial d^q_{x,x'}(t)$ to be between state pairs $(x,x')\in \mathcal{E}^q$, i.e., neighbors of state $x$ in ${\mathcal G}^q$.  Hybrid states can be characterized by switching conditions that depend on the evolution of the hybrid state dynamics hitting a boundary.

In this case, the load plasticity is in scheduling  transitions from hybrid state $x$ to hybrid state $x'$ before or after the boundary is hit. If the automaton can switch the discrete state every $\delta T$ and the injection rate is constant $G^q$, irrespective of the underlying job done in the state, the GHS model will end up resembling the RIC plasticity. We do not further elaborate for lack of space and worthwhile examples to explore.

\section{Using Cluster-based Load Plasticities for Planning and Control}\label{sec.why}

As mentioned previously, the primary contribution of this paper is in mathematically capturing the inherent flexibility of large appliance populations in a scalable fashion, which we addressed in Section \ref{pop.gen}.  However, the real value of such models is realized when they are coupled with DR functionalities. Thus, even though it is not the focus of this paper, to put our contribution into perspective, we discuss briefly how load plasticities can be used to design DR schemes.

As mentioned in the Introduction, hybrid stochastic load models can help planning and real-time control decisions of an aggregator. For example, in a  two-settlement structure:

{\bf Ex-ante}: the aggregator plans how much power $B(t)$ to purchase and how much ancillary service  capacity $M(t)$ to offer in the forward market. The aggregator picks $(B(t),M(t))$  to minimize an expected cost ${\mathcal C}^{F}(L(t),B(t),M(t))$, which includes costs and benefits of buying energy and selling ancillary services over a certain time horizon $\Omega$, i.e. 
\begin{equation}\label{ex-ante}
\min_{B(t),M(t)} \sum_{t\in \Omega}{\mathbb E}\{{\mathcal C}^{F}(L(t),B(t),M(t))\}~~\mbox{s.t.}~ L(t)\in { \mathcal L}^{DR}(t),
\end{equation}
where ${\mathcal L}^{DR}(t) \subseteq {\mathcal L}(t)$ denotes the set of possible load shapes that can be extracted from a population with load plasticity ${\mathcal L}(t)$ under a specific demand response strategy exercised by the aggregator. The expected value is over all future random realizations of load flexibility in ${ \mathcal L}^{DR}(t) $, given the random nature of the appliance arrivals, captured through $a^q_x(t)$ in our model. For brevity, we assume that $(B(t),M(t))$ will be cleared by the market operator at the marginal price.

{\bf Real-time}: the aggregator  is committed to control $L(t)$ to follow the schedule $(B(t),M(t))$ for the current time $t$ 
and minimize its real-time cost. The aggregator can be myopic:
\begin{equation}\label{myopic}
\min_{L(t)} ~{\mathcal C}^R(L(t),B(t),M(t))~~\mbox{s.t.}~ L(t)\in { \mathcal L}^{DR}(t),
\end{equation}
or be foresighted, solving a model-predictive problem similar to \eqref{ex-ante} with the real-time cost ${\mathcal C}^R$ as opposed to ${\mathcal C}^{F}$. At time $t$, the cost includes $H$ dummy future decisions for a time horizon $\Omega=\{t,t+1,\ldots,t+H\}$, constantly revised and updated in real-time. 
Note that such model-predictive strategies will be very hard to implement if the aggregator considers the continuous characteristics of every single load in the population. Thus, without clustering, the aggregator may need to resort to myopic policies for online control, the consequences of which are studied in the numerical example (Section \ref{sec.numerical}), wherein we solve \eqref{ex-ante} and \eqref{myopic}  to showcase the usefulness of our proposed paradigm.

Next we discuss the set ${ \mathcal L}^{DR}(t)$ for the two most prominently discussed types of continuous end-use demand management, i.e., direct load scheduling and dynamic retail pricing.

\subsection{Direct Load Scheduling (DLS)}\label{drtypes} For this type of demand management, the aggregator directly observes the arrivals of appliances in the grid, i.e.,   $\partial a^q_x(t)$ is known. In real-time, a control center makes scheduling decisions $\partial d^q_x(t)$ for the population present in the system, and relays these decisions to individual appliances. This makes DLS the tightest form of load management possible, as it can allow the aggregator to realize any possible load shape that the population can take. This amounts to
\begin{equation}{\mathcal L}^{DR}(t) = {\mathcal L}(t).\end{equation}

Of course, there is an economic side to DLS that we are not focusing on in this paper. The plasticity ${\mathcal L}(t)$ available to the aggregator depends on the choice of end-use customers to participate in the DLS program and allow the aggregator to schedule their consumption. This modifies the arrival rates of flexible loads, i.e., $a^q_x(t)$. For an aggressive and day-to-day type of direct load management, a rational customer will not provide this service for free. The design of appropriate economic incentives for direct load management is the subject of ongoing research, see, e.g., \cite{alizadeh, bitar2012deadline,kefayati}. In this paper, however, we do not address this problem and take  participation under a set of pre-specified customer preferences as given.

\subsection{Dynamic Pricing} In this type of demand management, time-varying retail prices constitute the only control knob available to the aggregator for both load control and billing. Here, ${\mathcal L}^{DR}(t)$ might not necessarily be equal to the inherent load plasticity of the population  ${\mathcal L}(t)$. Denoting the vector of retail prices posted at time $t$ for the look-ahead horizon $1,\ldots,\tau$ as $\mathbf{p} = [\pi^r(1),\ldots, \pi^r(\tau)]$, pricing can lead to the following set of possible load shapes:
\begin{align}
{\mathcal L}^{DR}(t) = \big\{ L(t) |& L(t) = f(t;\mathbf{p}(t)), \mathbf{p}(t) \in \mathcal{Z}(t) \big\}
\end{align}
where $\mathcal{Z}(t)$ is the set of all possible retail prices (probably partially regulated). The function $f(.)$ denotes the price-response of the population, which is equal to: 
\begin{align}\label{priceresp}
 f(t;\mathbf{p}(t)) = \mbox{argmin}_{L(t) \in \mathcal{L}(t)} \sum_{t = 1}^{\tau} \pi(t) L(t).
\end{align}

Notice the dependence of the price-response function on the  the inherent plasticity of the population  $\mathcal{L}(t)$ given by \eqref{totplas}. As opposed to the DLS case,  $\mathcal{L}(t)$  is not known to the aggregator here  since the arrivals $a^q_x(t)$ are not observable.
Consequently,

 1) In this case, the arrivals $a^q_x(t)$ need to be estimated to build the population model $\mathcal{L}(t)$ and the price-response. This would mean that the price-response can be treated as a parametric function with unknown variables $a^q_x(t)$ learned sequentially in time as new prices are posted every day. The aggregator can then use \eqref{priceresp} to estimate the price-response. 

 2) Another layer of complexity arises in the price-response learning problem due to the fact that demand flexibility may not be lossless. This means that customers may increase/decrease their total daily demand due to varying prices. Consequently, the arrival $a^q_x(t)$ can be a function of $\mathbf{p}(t)$. 

An alternative approach to sequential learning of $ f(t;\mathbf{p}(t))$ would be to have local third-party (non-interested) entities that analyze meter data through non-intrusive load monitoring techniques  \cite{nlm}. They can provide estimates of $a^q_x(t)$ to the aggregator. Non-intrusive load monitoring aims at disaggregating the consumption of different appliances in a household based on meter information.

Due to these extra layers of complexity, and lack of real-world data that captures the response of end-use customers to price,  in our numerical test cases we focus on DLS approaches.

\section{Scalable and Anonymous Telemetry}\label{sec.anoncom}
Fig. \ref{DLSchart} shows the steps that the aggregator should take to perform DLS. As the reader can see, the aggregator will need two-way communication with participating customers. A reader not interested in anonymous communication and submetering can simply skip this Section.

\begin{figure}[t]
\centering
\includegraphics[width = 0.8 \linewidth]{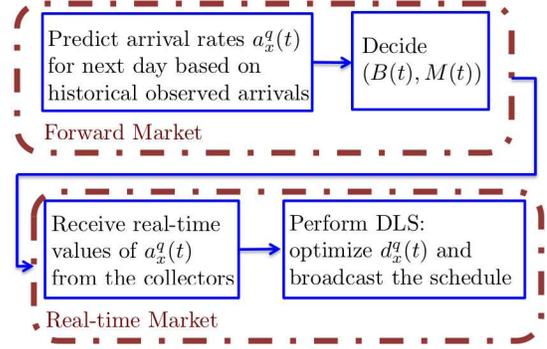} 
\caption{Control steps taken by aggregator performing DLS}
\label{DLSchart}
\end{figure}


We assume that the aggregator has a portfolio of $V$ appliance categories. We have discussed four possible categories in Sections \ref{cat1} to \ref{cat4}, but more categories could be possible as described in Section \ref{cat5}. The aggregator  is able to control large populations 
${\mathcal P}^v (v=1,\ldots,V$) of each category using their respective population models ${\mathcal L}^v(t)$. Next, we clarify how the communications with this population of devices can occur anonymously under DLS.  For our purposes, retaining anonymity means that the aggregator can be blind to the customer identity and still directly schedule their appliances.   

We note that a Dynamic Pricing scheme can be supported by the current AMI setup. In the uplink, household meter information would be processed to produce estimates of $\partial a_x^{q,v}(t)$ \cite{nlm}, while the control is merely a price signal. Obfuscating meter information to protect privacy, thereby corrupting the estimates of $\partial a_x^{q,v}(t)$, will have direct consequences on the control reliability.

The components of the telemetry architecture for DLS are:
\begin{enumerate}
\item \textit{Uplink  traffic}: The value of the arrival process $a^{q,v}(t)$ for each state $x$, cluster $q$ and category of appliance $v$ can be communicated anonymously to the aggregator; 
\item  \textit{Downlink  traffic}: To activate appliances, information regarding the decided schedule, e.g., $d^{q,v}(t)$ should be relayed to each waiting appliance in cluster $q$ in an appliance-blind fashion;  
\item {\it Measurement and verification}: A slower channel is necessary to measure that the DLS control actions are correctly executed by the recruited appliances that received participation incentives. 
\end{enumerate}

\subsubsection{Downlink} First, we shall clarify how the {\it downlink} works for DLS, i.e., how appliances are made aware of scheduling decisions. We omit the category index $v$ for brevity here.  Since the aggregator's scheduling decisions do not involve any individual appliance identifiers, the downlink communication should also be designed accordingly to function in an appliance-blind fashion. Next, we discuss our proposed solution for each of the four appliance categories discussed in Sections \ref{cat1} to \ref{cat4}.

$\bullet${\it Categories \ref{cat1} and \ref{cat2}}: Here appliances in the same cluster $q$ all share the same deadline $\chi^q$.   Thus, if the aggregator decides to move $d_{x,x'}^q(t)$ from state $x$ to state $x'$, it can simply broadcast the following list of ratios to all appliances in state $x$:
\begin{equation}
{\kappa}_{x,x'}^q(t) = \frac{d_{x,x'}^q(t)}{n_x^q(t)},~~x'\in \mathcal S_x
\end{equation}
Upon receiving these values, if an appliance happens to be in cluster $q$ at state $x$, they have to move to state $x'$ with probability ${\kappa}_{x,x'}^q(t)$. Given a large enough population size, this randomized scheme will perform well in executing the scheduling decisions.

$\bullet${\it Category \ref{cat3}}: The decision that needs to be broadcast is  $n_{x,1}^q(t)$, i.e., how many TCLs in cluster $q$ at state $x$ can be on at time $t$. Given the current value of $n_x^q(t)$, this activation can  be executed using a simple randomized policy with
\begin{equation}
{\kappa}_{x}^q(t) = \frac{n_{x,1}^q(t)}{n_x^q(t)}
\end{equation}
$\bullet${\it Category \ref{cat4}}:
Here we assumed that  the population in cluster $q$ share the same slack time value $\chi^q$. However, since some appliances arrive earlier than others and should naturally be scheduled earlier, we can assume that arriving appliances are queued in a first-come-first-out (FIFO) discipline, i.e., if you arrived earlier in cluster $q$, you are scheduled earlier.  
But how can the aggregator address appliances with higher priorities in the queue when their identity is unknown? 
The idea is simple.  At time $t$, the aggregator needs to find the time epoch $\tau^q \leq t$ at which the number of appliances that arrived in cluster $q$ since the origin of time, i.e., $a^q(\tau^q)$,  was marginally higher than the number of appliances the aggregator plans to have activated since time epoch $1$ until time $t$, i.e., $d^q(t)$:
\begin{equation}\label{alpha2}
\tau^q(t)=\min\lbrace t' \leq t: a^q(t') \geq d^q(t)\rbrace.
\end{equation} 
If this time index $\tau^q(t)$ is broadcast to \textit{all appliances} in cluster $q$, then the appliances that  submitted requests before time $\tau^q(t)$   can start their task while the rest wait until authorized. If the difference between $a^q(t')$ and $d^q(t)$ is not negligible, a randomization ratio
\begin{equation}
{\kappa}^q(t) = \frac{ a^q(t') - d^q(t)}{\partial a^q(t')}
\end{equation}
can also be broadcast along with $\tau^q(t)$ to ensure that nearly $d^q(t)$ appliances are turned on.

\subsubsection{Uplink}
In the {\it uplink},  the aggregator just needs to gather the values of the arrival processes  $a^{q,v}_x(t)$\footnote{For TCLs, the switch process $D^{q,v}_{x,x'}(t)$ is communicated as well.} from the population. These values are aggregated across customers, but disaggregated across consumption cluster $q$, state $x$ and appliance categories $v$. In other words, 
while it is true that the appliances' specific use modality is visible, no explicit information on appliance owners is required for the aggregator's decisions. This makes any DR scheme using this model anonymous.

Each arrival event, i.e. every instance of $\partial a_i(t)=1$, is marked by the transmission of a packet, including the appliance type $v_i$, its appropriate cluster index $q_i$ (associated to $\boldsymbol{\theta}^q$), and potentially the state $x_i$. This results in a short packet, and can be easily integrated with simple application layer protocols used in industrial control and AMI networks. 
Because of this, the total communication traffic, in terms of number of average packets per $\delta T$ that enter the network, scales exactly as  $\sum_{x,q,v}\mathbb{E}\{\partial a^{q,v}_x(t)\}$. In each $\delta T$, the probability of transmission of a packet per appliance is  a Bernoulli random variable with success probability equal to $p_i(t)$, i.e.:
\begin{equation}\label{rate}
p_i(t)=\mathbb{E}\{\partial a_i(t)\}.
 \end{equation}

Note that the the discrete time bin $\delta T$ needs to be greater than the network end to end delay. This is the delay of getting a request forwarded from a HEMS upstream to the point where final delivery to the aggregator is made. This delay will vary widely depending on whether the AMI network or a conventional network provider serves this traffic. Current AMI mesh network solutions have significant latency but transmit a long record of raw samples of meter data. This is a much heavier task compared to forwarding the payload $(v_i,q_i,x_i)$ associated with the appliances' requests.   
For instance, IEEE 802.15.4 (Zigbee), considering the packet acknowledgment, has a latency per packet bounded below by about $\delta \tau= 7$msec. However, this is assuming that the channel is idle. Denoting the ratio $\delta \tau/\delta T$ by $\alpha$, in a certain neighborhood $\mathcal N$, the packet throughput per second is 
\begin{equation}\label{ratemac}\rho_{\mathcal N}(t)=\delta T^{-1}\sum_{i\in {\mathcal N}}p_i(t)\prod_{j\in {\mathcal N}/i}(1-\alpha p_j(t)).\end{equation}

To ensure that the data presented to the aggregator can be anonymous, a trusted third party has to act as an intermediary, placing neighborhood \emph{collector} stations throughout the population.  These third-party collectors manage their own neighborhood AMIs and tally each received payload $(v_i,q_i,x_i)$ to compute a local $a_{x}^{q,v}(t)$ that is then sent to the aggregator.

While collectors hide the identities of participating customers, their updates on $a_{x}^{q,v}(t)$ reveal some information about their possibly small neighborhoods.  To mitigate this, we assume that the collectors have access to an out-of-band channel that can be used to deliver updates to the aggregator.  Using this channel, anonymous routing protocols (e.g., onion routing between collectors \cite{Dingledine2004TSO1251375.1251396}) can obfuscate the true sender of a message, allowing anonymous delivery of updates.

Note that the presence of neighborhood collectors is necessary to preserve customer anonymity when an otherwise significant portion of the AMI would be owned and operated by the aggregator.  As we will explain in Section \ref{sec.mv}, collectors also have the added benefit of providing verification mechanisms for DLS and align well with our goals.

 While with these check and balances, it is possible to ensure {\it anonymous} transactions, 
 the question that remains is whether anonymity is sufficient to ensure privacy. 
The problem that should be investigated is whether knowing aggregated $a_x^{q,v}(t)$ across the population can reveal information about individual users.  This can be viewed a special case  of the broader problem posed in {\it differential privacy} papers \cite{lalitha,diff2}. We leave the study of this topic to future work.

\subsubsection{Measurement \& Verification (M\&V) for DLS}\label{sec.mv}

Naturally, if DLS is performed, a verification mechanism must be in place to ensure that the control action takes place and economic benefits or penalties for the customer are appropriately accounted for.  Furthermore, the customer can have real-time reports signaling if something goes wrong with the flexible appliance control system. The aggregator can also receive appropriately filtered reports to monitor the status and bill its appliance fleet. 

During the day, the customer will have a specific household daily load profile $\boldsymbol{L}_i=(L_i(0),\ldots,L_i(N))$, where $N\delta T$ is equal to 
the $24$h period. This vector will have two components, 
\begin{equation}
\boldsymbol{L}_i=\boldsymbol{L}^I_i+\boldsymbol{L}^F_i 
\end{equation} 
one inflexible and random part $\boldsymbol{L}^I_i$, and another flexible portion $\boldsymbol{L}_i^F$ that has to follow the aggregator's dispatch in a DLS program. If the control took place correctly, $\boldsymbol{L}_i^F$ should be entirely predictable ex-post based on the control actions. More precisely, $\boldsymbol{L}_i^F$ needs to be equal to a dispatched value $\boldsymbol{L}_i^{DR}$, whose value is known to the collector.
The most precise approach to verify that $\boldsymbol{L}_i^F$ is equal to $\boldsymbol{L}_i^{DR}$ is to sub-meter the controllable appliances through a  trusted Application Programming Interface (API) in the HEMS. The API can then forward the submetered load to the collector.

Theoretically, submetering can be replaced by a strategy that allows the collector to search for the presence of the desired profile in the aggregate household consumption.  The M\&V detection problem can be then posed as the detection of the expected load component $\boldsymbol{L}_i^{DR}$, which may include multiple controllable appliances, immersed in the {\it noise} produced by the uncontrollable random load $\boldsymbol{L}^I_i$.
The M\&V module needs to track the statistics of $\boldsymbol{L}^I_i$, $p_{\boldsymbol{L}^I_i}({\bf f})$\footnote{The statistics are likely to be seasonal but we ignore that for the moment to avoid complicating the notation.}. Then, it can determine if $\boldsymbol{L}_i^{DR}$ is present using the likelihood function
\begin{equation}
\log p_{\boldsymbol{L}^I_i}(\boldsymbol{L}_i-\boldsymbol{L}_i^{DR} )\geq \eta
\end{equation}
For instance, if $\boldsymbol{L}^I_i$ is a normal vector with mean $\boldsymbol{\mu}_{\boldsymbol{L}^I_i}$ and covariance 
$\boldsymbol{W}_{\boldsymbol{L}^I_i}$, the test is equivalent to 
$$\|\boldsymbol{L}_i-\boldsymbol{L}_i^{DR}-\boldsymbol{\mu}_{\boldsymbol{L}^I_i}\|^2_{\boldsymbol{W}_{\boldsymbol{L}^I_i}^{-1}}\leq \eta'.$$
This means that the wider the variance of $\boldsymbol{L}^I_i$ is compared to $\|\boldsymbol{L}_i^{DR}\|^2$, the more the decision is plagued by uncertainty. This could make it possible for the customer to deceive the aggregator, even without tampering with the sensors. However, since flexible appliances are likely to be the most significant components of the household consumption pattern, it is highly likely that the error will be limited.   

\vspace{-0.2cm}

\section{Numerical Case Study}\label{sec.numerical}
To showcase the benefits of our scalable framework for modeling and scheduling heterogeneous loads, we simulate the operation of an aggregator dealing with large populations of EVs and TCLs. As discussed in Section \ref{cluster.sec}, the operations of the aggregator can be divided into two steps: ex-ante planning to determine the optimal energy market purchase and ancillary service capacity to provide, and real-time operations to manage the load. For load management, we choose to perform DLS on both TCLs and EVs. As discussed in Section \ref{drtypes}, direct scheduling eliminates the need for modeling the customers' price response, which is not straightforward given the lack of real-world data.

\subsection{Non-interruptible Plug-in Hybrid  Electric Vehicles - Setting}
We test our proposed model on a simulated population of 40,000 Plug-in Hybrid EVs (PHEV) that receive home charging. The operation of the aggregator is simulated for one full day, and the effects of shifting the load forward to the next day is captured by extending the duration of this study to 32 hours. Participating EVs in the DLS program arrive randomly, according to a non-homogeneous Poisson arrival process. The arrival rate, as well as the SoC and tolerable delay of each request, is simulated using real-world statistics in \cite{alizadeh2013ev}. We divide arriving vehicles between 15 different clusters,  representing different discrete SoCs between 0 to 4 hours and slack times between 1 to 3 hours (full charge corresponds to 5 hours of charging from a zero SoC).  Charge requests are taken to be non-interruptible once started. We add the simulated charging load of this population on top of a typical daily base load scaled to have PHEV demand as $10\%$ of peak load\cite{data}. 

To numerically showcase the benefits of our scalabale yet accurate load model, we pick the aggregator objective to be a simple one: minimize the expected cost of serving the load given perfect forecasts of market prices but uncertainty in load arrival statistics $a^q(t)$.
The aggregator's decision structure follows that of a two-settlement energy market (Section \ref{sec.why}), which we describe next.

\subsection{Non-interruptible PHEVs - Ex-ante Planning}\label{sec.PHEVex-ante}
The ex-ante cost captures the cost of buying energy in the forward market, plus the expected cost in the real-time market. Each day includes $h = 1,\ldots, H$ hours and $\ell = 1,\ldots,T$ subhourly hours in each hour. To make our results interpretable to the reader, we assume that the aggregator has access to perfect predictions of the forward market price, denoted by $\pi^{F}(h)$, and real-time upward and downward clearing prices, denoted by $\pi^{R}_{up}((h-1)T+\ell)$ and $\pi^{R}_{dn}((h-1)T+\ell)$ respectively. However, this is not a necessary element for our model to work.

The optimization that should be solved ex-ante is
\begin{align}
\min_{B(h)} \sum_{h=1}^H & {\mathbb E} ~\{ {\mathcal C}^{F}(L(t);B(h))\} = \label{dacost}\\ 
\min_{B(h)} \sum_{h=1}^H &\bigg\{ \pi^{F}(h) B(h) +  {\mathbb E}\{ {\mathcal C}^{R}(L(t);B(h)) \}  \bigg\}\nonumber\\
\mbox{s.t.}~~~~~& L(t)\in { \mathcal L}(t),\nonumber
\end{align}
where,
\begin{align}\label{rtcost}
&{\mathcal C}^{R}(L(t);B(h)) =\\ &\sum_{h=1}^{H} \sum_{\ell = 1}^T \big\{ \pi^{R}_{up}((h-1)T+\ell) (L((h-1)T+\ell) - B(h))^+ \nonumber\\ &~~~~~~~~~~+\pi^{R}_{dn}((h-1)T+\ell) (L((h-1)T+\ell) - B(h))^- \big\}. \nonumber\end{align}

\begin{figure}
{\includegraphics[width = \linewidth]{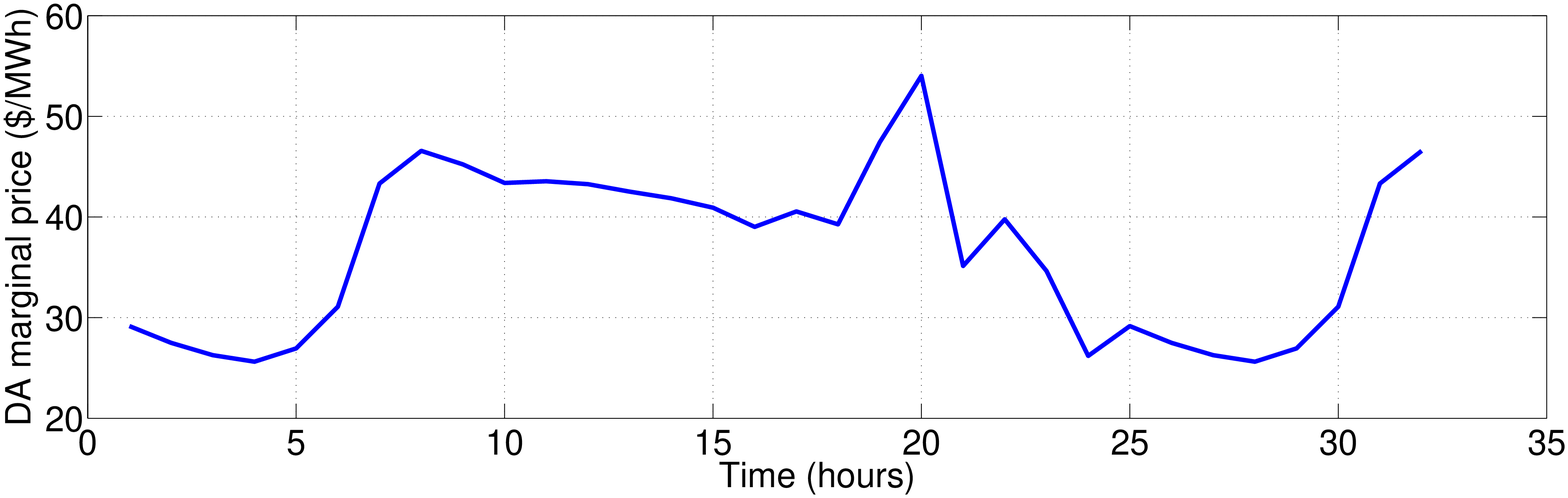}}
\caption{The day-ahead marginal price curve used in our case study }
\label{fig:price}
\end{figure}

Prices $\pi^{F}(h)$ reflect those of the PJM day-ahead energy market 10/22/2013				
 (see Fig. \ref{fig:price}). To only observe the effects that are of interest in this paper, we assume that the aggregator has access to perfect forecasts of  its uncontrollable load, and the only uncertainty is in the deferrable load. To avoid observing  effects of the aggregator trying to trade day-ahead prices against hour-ahead prices, we pick upward adjustment prices  to be 20$\%$ higher than day-ahead prices, and  downward adjustment prices to be 20$\%$ lower than day-ahead.

To solve the above mentioned ex-ante planning problem and determine the optimal $B(h)$, the aggregator needs a model of the load flexibility of the PHEV population on the next day.
To highlight the attraction of using a medium-grain model of non-interruptible load populations in \eqref{lint}, we compare the aggregator cost under our clustering method to that of a  tank model previously used in the literature as a population model \cite{homer, 6426102,ortega}. A tank model only takes the causality and deadline constraints for each appliance's energy consumption into account. Any other information about the specific shape of individual loads, i.e., uninterruptible charge at 1.1 kWs, is tossed out. Thus, the tank model will effectively model each non-interruptible PHEV as an ideal battery, and the population load plasticity will be captured by an extension of \eqref{occur4} to capture 32 different possible deadlines through 32 clusters:
  \begin{align} 
{\mathcal L}^v(h) = \Big\{L(h)|&  L(h) = \sum_{q=1}^{32} \sum_{x=0}^{5} \sum_{x'=0}^{5}(x'-x)\partial d^q_{x,x'}(h),\nonumber\\ &\partial d_{x,x'}^q(h) \in \mathbb{Z}^+,\!\! \sum_{x=0}^{5}\partial d_{x,x'}^q(h)\leq n^q_x(h),\nonumber \\&  n_x^q(q) = 0, \forall x < 5\Big\},\label{tank}
\end{align}
with
$n_{x}^q(h)=  a^q_x(h) + \sum_{x=0}^{5} [d^q_{x',x}(h-1)- d^q_{x,x'}(h-1)]
$.

Both population models \eqref{tank} and \eqref{lint} are inherently linear, allowing the cost-minimizing optimization \eqref{dacost} to be linear. Thus, given the stochastic nature of vehicle arrival patterns captured through $a^q_x(h)$, we can find $B(h)$ under both models by solving stochastic Integer Linear Programs (ILP). The good news is that, if needed, integrality constraints on the decision variables $d^q(h)$ and $d^q_{x',x}(h)$ can be relaxed at high aggregation levels without degrading performance.

The stochastic nature of vehicle arrivals in clusters is captured by a non-homogeneous Poisson process, which we show to be an appropriate choice based on real data in \cite{alizadeh2013ev}. Thus, we apply sample average approximation (SAA) to the random vehicle arrivals and compute the expected cost considering multiple scenarios for the vehicle arrival rates. Since the cost under any scenario is linear with $d^q(h)$, we solve a deterministic ILP. See the $B(h)$ determined under the tank model and our hybrid model in Figures  \ref{fig:tank} and \ref{fig:clustering} respectively. Notice that the $B(h)$ determined under the tank model is bouncier, due to the higher plasticity of ideal batteries. We will see why this is not a good choice for $B(h)$ next.

\subsection{Non-interruptible PHEVs - Online Scheduling}
We simulate the performance of a direct scheduler whose objective is to make real-time scheduling decisions ($d^q(t)$ in \eqref{lint}) to minimize \eqref{rtcost} in real-time, with different $B(h)$ given by the two models.  We approach the real-time scheduling problem in two ways:

\begin{figure}
\centering
\includegraphics[width = \linewidth]{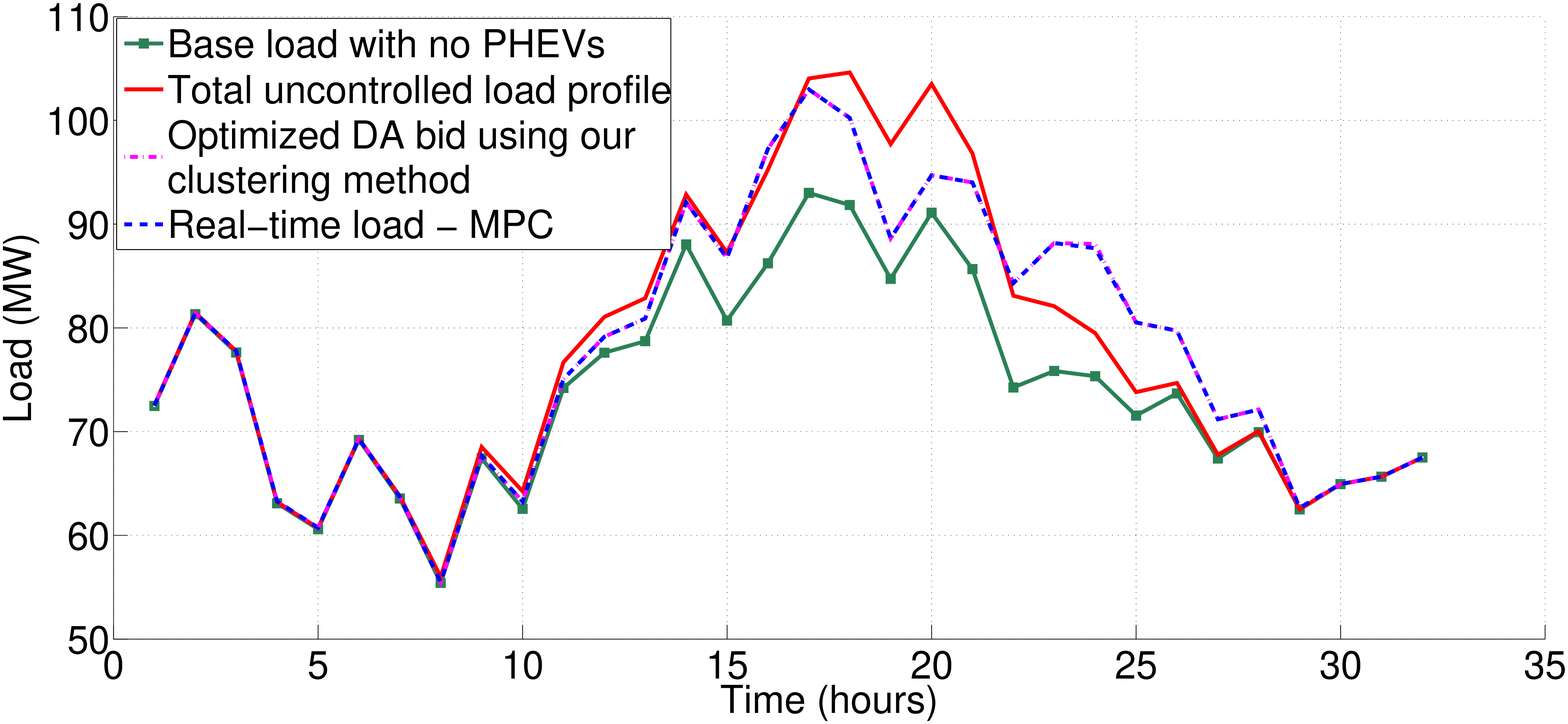} 
\caption{The aggregator can perfectly follow the day-ahead bid that was optimized under our proposed demand clustering method.}
\label{fig:clustering}
\end{figure}

\begin{figure}
\includegraphics[width = \linewidth]{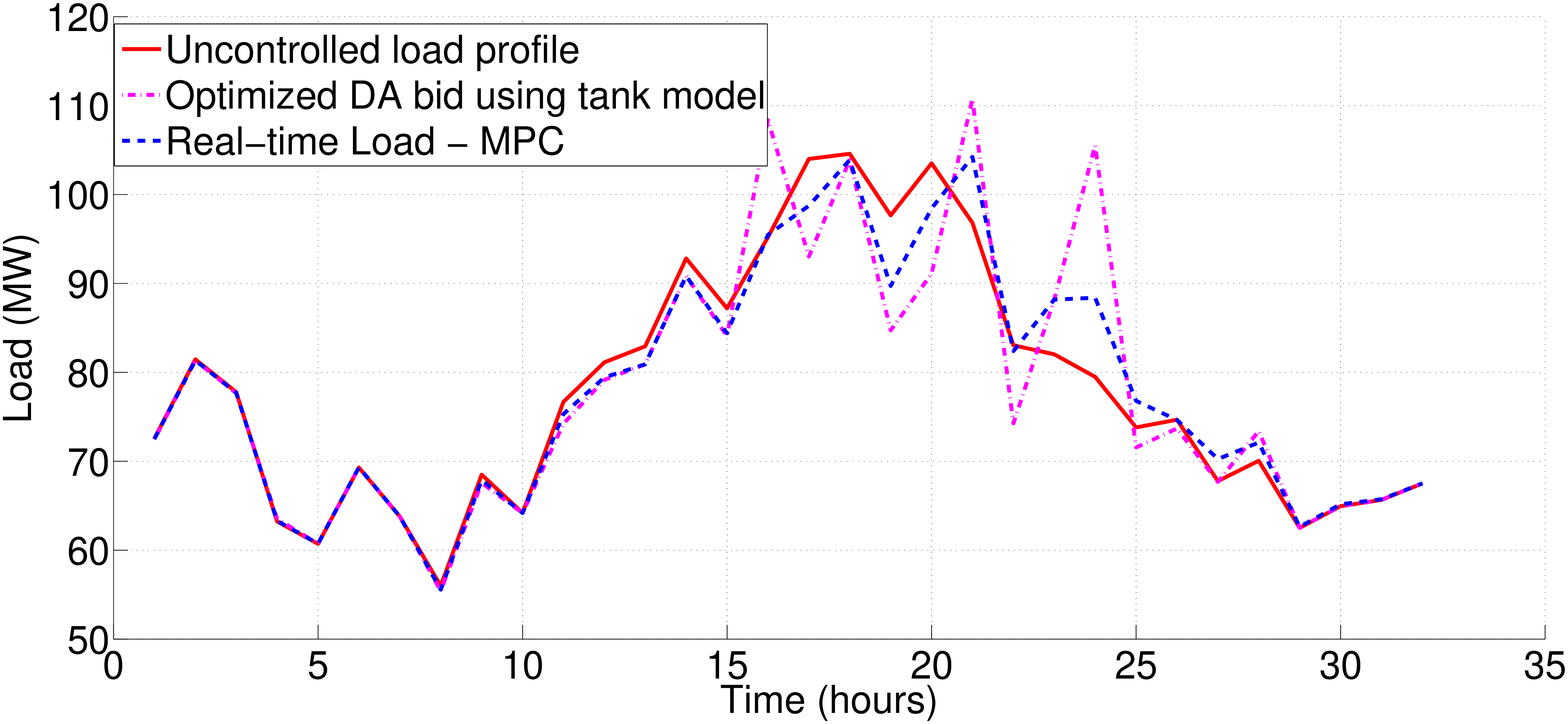} 
\caption{The aggregator cannot follow the day-ahead bid optimized using the tank model, simply due to modeling errors.}
\label{fig:tank}
\end{figure}

\subsubsection{Model-predictive Control (MPC) Based on Load Plasticity} Here we minimize the real-time cost in a model-predictive fashion, given the low complexity of solving such problems under our proposed clustered load model. 
Note that in real-time, the arrival counts of EVs in each cluster $q$, i.e., $a^q(h)$, are only causally observed. Thus, at each $h$, the scheduler has to solve a model predictive optimization to make scheduling decisions $d^q(t)$ while taking into account the possible scenarios for future arrivals. Here, we take a certainty equivalent approach at each $h$, replacing future arrival numbers of EVs for each cluster $q$  with their expected value $\mathbb{E}[a^q(h')] = \lambda^q(h'), h'>h$. See \cite{alizadeh2013ev} for a discussion of how $\lambda^q(h')$ can be predicted.

\subsubsection{Myopic Approach}Without load clustering, the problem of scheduling 40000 loads in real-time can be computationally prohibitive. Thus, here we simulate the performance of a myopic scheduler that, at each time $t$, simply picks the PHEVs that are closest to their deadline and schedules those vehicles for charge. The number of vehicles picked at each time is chosen such that the scheduled load is as close as possible to the forward purchase. We refer to this myopic scheme as Earliest Deadline First (EDF). Many authors have previously investigated EDF for EV scheduling, e.g., \cite{chen2012large,subramanian2012real}.

\begin{figure}
\centering
\includegraphics[width = \linewidth]{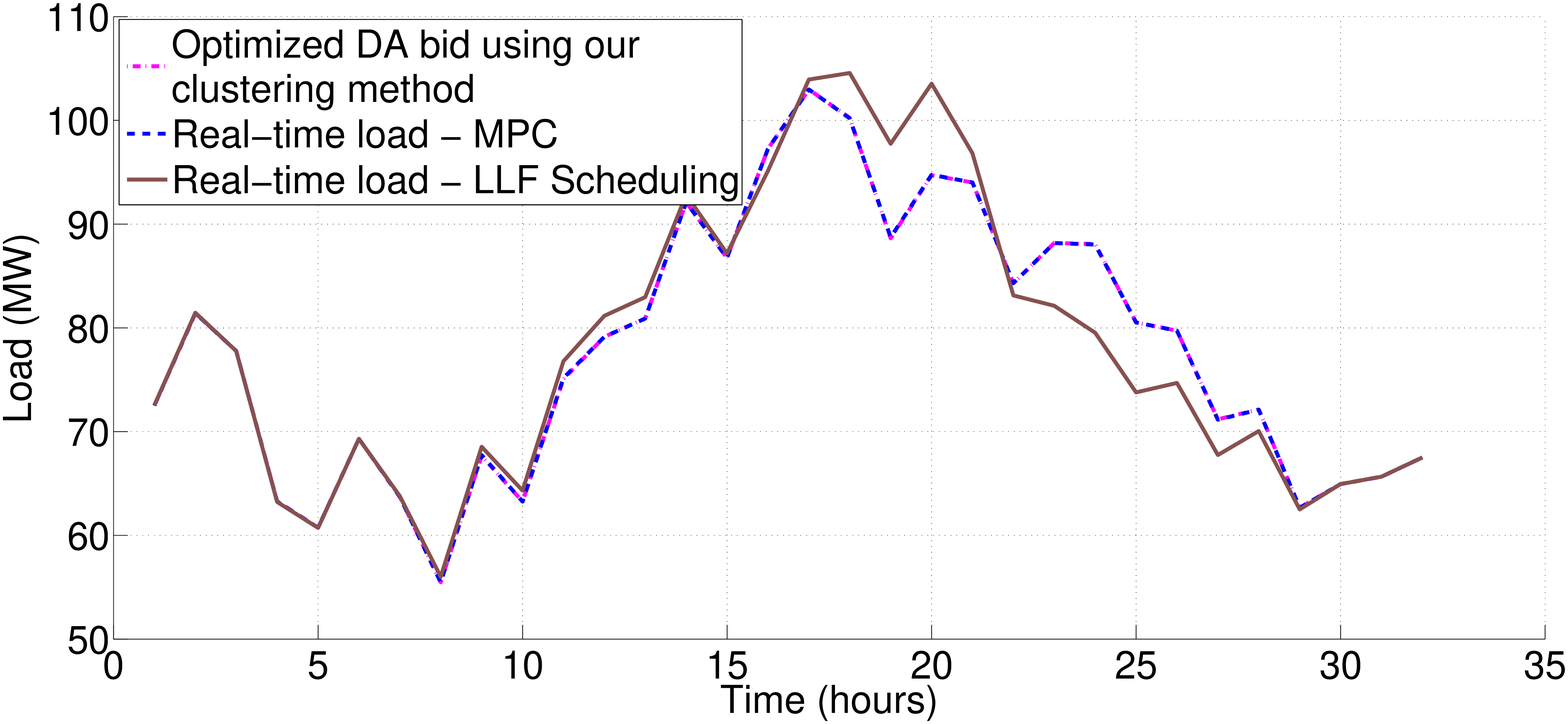} 
\caption{Performance comparison of LLF scheduler with the MPC scheduler}
\label{fig:llf}
\end{figure}

To compare the performance of population models \eqref{tank} and \eqref{lint}, we use the MPC method to schedule load in real-time, with the goal of minimizing \eqref{rtcost} under the forward purchase determined from these two models.
Figures  \ref{fig:clustering} and \ref{fig:tank} compare the actual hourly loads of the aggregator in these two scenarios. The reader can observe that the coarse nature of the tank model leads to large deviations of the real-time load from the bulk purchase. With the clustering method, on the other hand,  the deviation of the real-time load from the bulk purchase is nearly negligible. 
To give an intuitive explanation for the poor performance of the tank model, the reader should note that non-interruptible PHEVs have a more restricted load plasticity than the ideal battery.

To compare the effectiveness of MPC to that of LLF, we simulate the performance of both real-time scheduling schemes over the forward purchase determined through our model in Fig. \ref{fig:llf}. Notice the considerable deviations from the forward purchase for the LLF scheduler.
%

\subsection{Thermostically Controlled Loads - Setting}
Another useful property of scalable hybrid models is in determining the ancillary service capacity that a large population of flexible appliances can collectively provide. 
Here we simulate the performance of an aggregator that directly controls the on/off state of 10000 residential TCLs in heating mode for a duration of 6 hours. The ambient temperature for each hour $h$, $x_{{\mathrm amb}}(h)$, in the simulated interval mirrors that of Jan 29th, 2012 in Davis, California between 12 to 6 am. 

The comfort band, wattage $G_i$, initial temperature $x_i^{\rm init}$, loss rate $k_i$, and initial state $b_i$ of each device were generated randomly. The parameters $G_i$ and $k_i$ were randomly chosen to be independent and identically distributed (iid) from  uniform distributions $U([2000,4000])$ Btu/h and $U([50,200])$ W/C, and the heat gain was simply picked to be equal to $G_i$ Watts. The desired temperature of each device, $x^*_i$, came from an iid uniform distribution $U([69,75])^{\circ}F$, and $B_i$ from $U([2,4])^{\circ}F$. 

Remember that for population models, we needed to quantize the individual parameters of each device and assign them to clusters.  Here, $G_i, k_i, x^*_i$ and $B_i$ were quantized into 3, 4, 4, and 2 levels respectively. Thus, the TCLs were divided into 96 clusters to form the population model. For simulation purposes, the initial temperature was uniformly chosen to be in the comfort band, while initial on/off states were chosen from a Bernoulli distribution with success probability of $s = 0.15$. 

In what we refer to
as the autonomous switching scheme, a heating unit is turned
on when the temperature dips below the band and turned
off when it gets above the band. This describes the normal switching behavior of a TCL. The resulting load is what we take as the {\it baseline} load, which has the following expected value in the forward market:
\begin{eqnarray}
{\mathbb E}[L^{\rm base}(h)] = \sum_{q=1}^Q n^q(h) G^q\frac{ \tau_{\rm on}^{q}(h)}{\tau_{\rm on}^{q}(h) + \tau_{\rm off}^{q}(h)},\\
n^q(h)={\mathbb E}\left[\sum_{x\in {\mathcal S}} \sum_{\ell=1}^T n^q_x((h-1)T + \ell)\right],
\end{eqnarray}
and
 $\tau_{\rm on}^{q}$ and $\tau_{\rm off}^{q}$ respectively denote the expected on and off times of a TCL in cluster $q$ on the day-ahead. They are:
 \begin{align}
 \tau_{\rm on}^{q}(h)&=\frac{1}{k^q} \ln\left(
\frac{x^{*q}-\frac{B^q} 2 -\frac{G^q}{k^q}-x_{\rm amb}^q(h)}
{x^{*q}+\frac{B^q} 2-\frac{G^q}{k^q}-x_{\rm amb}^q(h)}
\right),\\
 \tau_{\rm off}^{q}(h)&=\frac{1}{k^q} \ln\left(
\frac{x^{*q}+\frac{B^q} 2 -x_{\rm amb}^q(h)}
{x^{*q}-\frac{B^q} 2 -x_{\rm amb}^q(h)}
\right).
 \end{align} 
 To simplify the compilation of these these average values, we take the noise term $\alpha_i(t)$ in \eqref{tcls} to be equal to its expected value, ${\mathbb E}[\alpha_i(t)]= x_{{\mathrm amb}}(h) a^q $ (c.f. \cite{globalSIP}).

\subsection{Thermostatically Controlled Loads - Ex-ante Planning}\label{sec.TCLex-ante}\label{sec.capacity}
Ex-ante, the objective of the aggregator is to find the maximum regulation capacity that it can safely provide using a population. In real-time, the status of the TCL devices has to be scheduled to follow the dispatch as closely as possible. Here we are focused on the ex-ante modeling. 

The AGC signal determines how much the TCL aggregate load is required to deviate from the baseline load. 
AGC signals  are claimed to resemble independent random variables with a zero mean. In practice, this is not the case and the signals are correlated in time. Hence, to be reliable, the aggregator must be able to increase/decrease demand from the baseline and hold the demand at that value for a certain duration.
The natural question that comes to mind is: how much can the demand be increased/decreased from the baseline? For TCLs in different clusters, the answer is different. 

Let us discuss the extreme case in which a TCL has to increase/decrease its load from the baseline in a stationary setting (i.e. for an indefinitely long duration). In this case, the change in consumption from base load is no longer considered as a transient behavior. This stationary load can be thought of as the load of a TCL that autonomously switches   in a comfort band $[\theta^{\max}_i,\theta^{\max}_i] \subset [x^*_i-B_i/2,x_i^*+B_i/2]$.
There is a limit to the highest and lowest amount of stationary consumption that can be achieved by a population of TCLs. If we choose to have a minimum width of $\Theta$ units for the comfort band, the scenario corresponding to the maximum amount of stationary power consumption for cluster $q$ is $\theta_{\max}^q = x^{*q}+B^q/2$ and $\theta_{\min}^q= x^{*q}+B^q/2 - \Theta$. This is because, in this case, the temperature respectively drops and increases at its quickest and slowest possible rate (see $\partial x_i(t) $  in \eqref{tcls}). This leads to the maximum possible amount of time that the TCL is on. Since TCLs,  when on, have constant power consumption $G^q$, the maximum stationary expected load of the population is:
\begin{eqnarray}
{\mathbb E}[L^{\rm max}(h)] = \sum_{q=1}^Q n^q(h) G^q\frac{ \tau_{\rm on}^{q,\max}(h)}{\tau_{\rm on}^{q,\max}(h) + \tau_{\rm off}^{q,\min}(h)},\end{eqnarray}
where
 $\tau_{\rm on}^{q,\max}$ and $\tau_{\rm off}^{q,\min}$ respectively denote the expected on and off times of a TCL in cluster $q$  operating in the comfort band $[x^{q*}+B^q/2 - \Theta, x^{*q}+B^q/2]$. They are:
 \begin{align}
 \tau_{\rm on}^{q,\max}(h)&=\frac{1}{k^q} \ln\left(
\frac{x^{*q}+\frac{B^q} 2 - \Theta-\frac{G^q}{k^q}-x_{\rm amb}^q(h)}
{x^{*q}+\frac{B^q} 2-\frac{G^q}{k^q}-x_{\rm amb}^q(h)}
\right),\\
 \tau_{\rm off}^{q,\min}(h)&=\frac{1}{k^q} \ln\left(
\frac{x^{*q}+\frac{B^q} 2 -x_{\rm amb}^q(h)}
{x^{*q}+\frac{B^q} 2 - \Theta-x_{\rm amb}^q(h)}
\right).
 \end{align} 
Similarly, 
\begin{eqnarray}
{\mathbb E}[L^{\rm min}(h)] = \sum_{q=1}^Q n^q(h) G^q\frac{ \tau_{\rm on}^{q,\min}(h)}{\tau_{\rm on}^{q,\min}(h) + \tau_{\rm off}^{q,\max}(h)},
\end{eqnarray}
where $\tau_{\rm on}^{q,\min}$ and $\tau_{\rm off}^{q,\max}$ respectively denote the average on and off times of a TCL in cluster $q$  operating in the comfort band $[x_i^*-B_i/2, x_i^*-B_i/2+\Theta]$. 

This leads to a very conservative estimate of the regulation capacity that can be safely provided by the TCL population:
\begin{eqnarray}
M \!\!=\!\! \min_h \min{\mathbb E}\{L^{\rm base}(h) \!-\! L^{\rm min}(h),  L^{\rm max}(h) \!-\!  L^{\rm base}(h)\}.
\end{eqnarray}

In reality, to provide regulation services, the TCLs do not have to hold their power consumption at $L^{\rm max}(h)$ and $L^{\rm min}(h)$ for a very long time. One way to get a less conservative and more profitable estimate of the capacity, $M'$, is to simulate the response of each cluster to a control signal that increases or decreases the load from its baseline by a certain step of variable height $m$. The population should be able to increase and hold the consumption at baseline plus $m$ for a certain amount of time $\upsilon$. Here we choose $\upsilon$ to be the 97 $\%$ quantile of the zero-crossing time of historical regulation signals. The maximum step size $m$ for which the response has an acceptable error is $M'$. When simulating the response, the initial temperature of all units is chosen uniformly at random within their comfort band. Here we take a response to be acceptable if the deviation from the target load is less than $0.05 M'$ for 95 percent of the duration $\upsilon$.

 The value of $\upsilon$ is  19 minutes based on 18 days of regulation signal  data available through PJM. By simulating the response of each cluster to different step sizes at different hours $h$, we determine the value of maximum safe deviation from base load for each $q$, which we dub as $M^q(h)$. Consequently,
\begin{eqnarray}
M' =  \sum_{q=1}^Q  \min_{h}n^q(h) M^q(h),
\end{eqnarray}
which is equal to 2.05 MWs for our simulated TCL population.

\subsection{Thermostatically Controlled Loads - Online Scheduling}\label{sec.TCLonline}

 In real-time, the status of the TCL devices has to be directly scheduled so that the load follows the regulation dispatch signal as closely as possible. We use the coarse clustering method in \eqref{coarsecluster} to perform this scheduling. Clusters are assigned to different deadlines one minute apart within the next 30 minutes, resulting in 60 separate clusters for both on and off TCLs. After each switching event, each TCL waits for a short time (courtesy period). Then, they submit their next switching deadline to the scheduler and wait. If anytime during this wait the temperature goes out of the comfort band, the TCL autonomously switches on/off. These are the events that can result in deviations from the dispatch, and are mainly due to forecast errors in the deadlines.

\begin{figure}[h!]
  \centering
      \includegraphics[width=\linewidth]{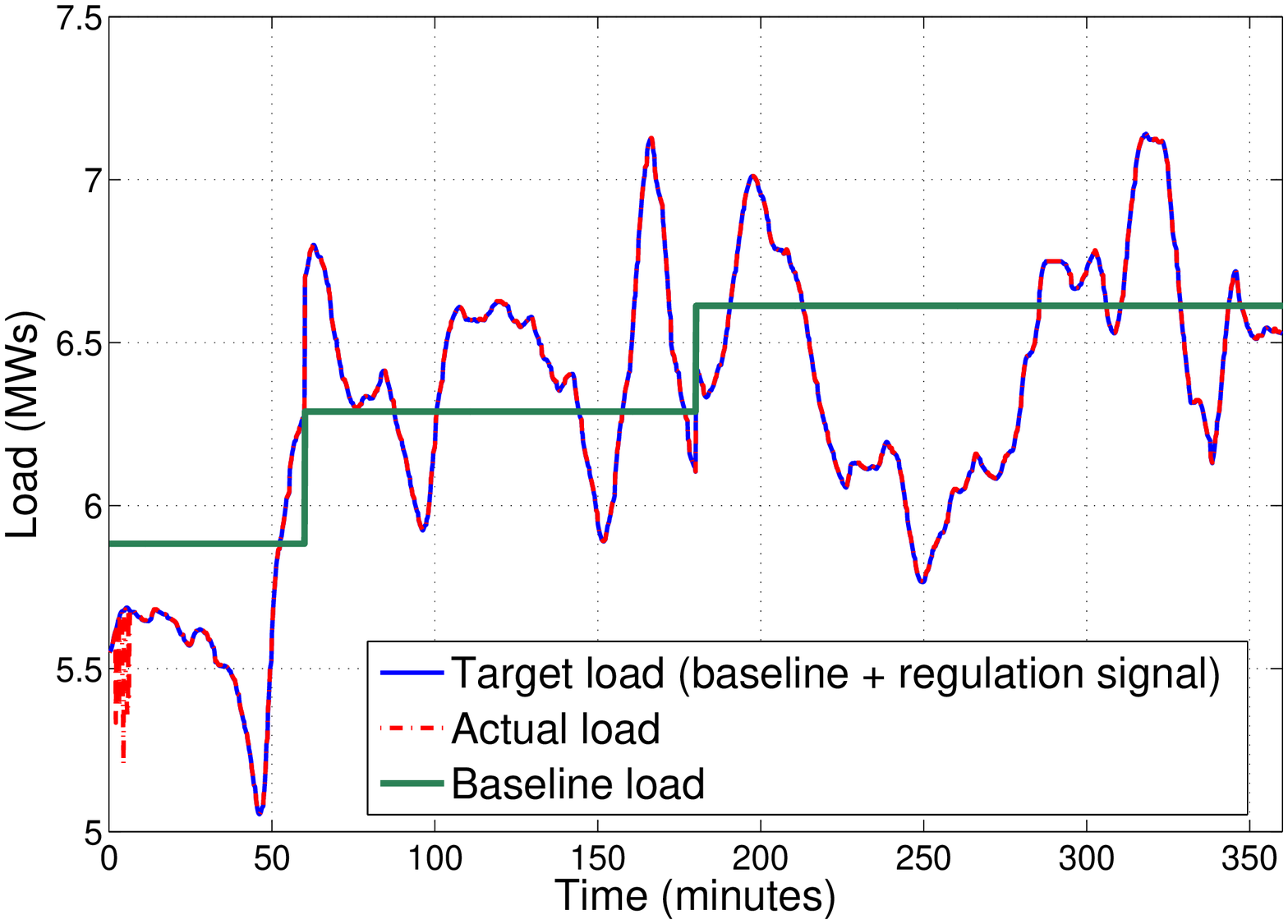}
\includegraphics[width=\linewidth, height=3.2cm]{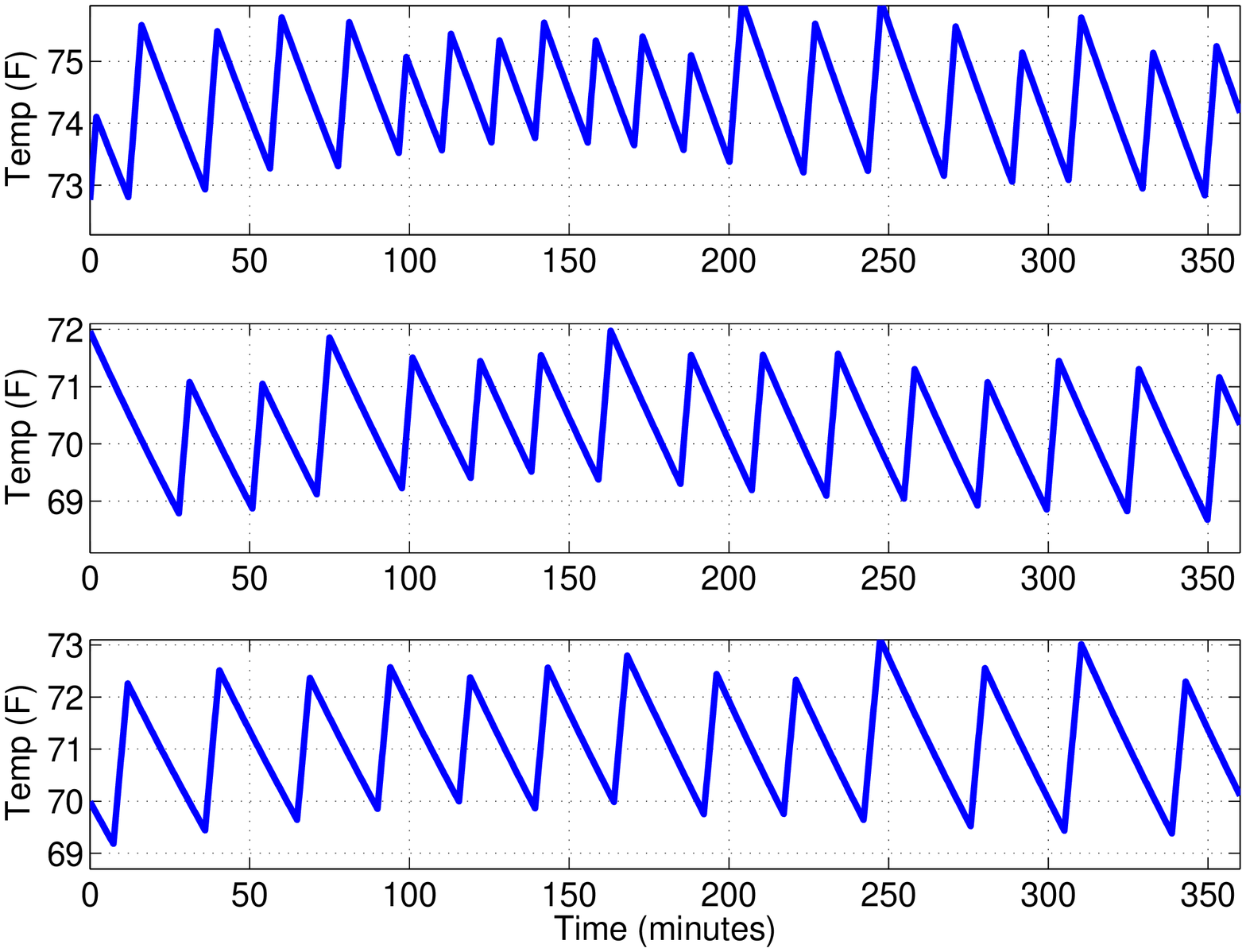}
  \caption{Top: Simulated response of the TCL population to regulation signals. Bottom: Representative temperature of three 2 ton A/C units. The y-axis range is the same as the unit comfort band.}
  \vspace{-.5cm}
\label{tclresponse10000}
\end{figure}

Normalized regulation signals available online through PJM were used, scaled to the regulation capacity offered by the aggregator.
Fig. \ref{tclresponse10000} shows the performance of our proposed algorithm for a duration of 6 hours. The aggregate consumption follows the dispatched deviation from the baseline consumption very closely. Individual representative temperatures of 3 random participating units is shown concurrently. We believe that the choice of $M'$ was near-optimal, since the population load response was not acceptable when the regulation capacity was increased to 2.4 MWs.

\subsection{Communication requirements}\label{sec.tp}
Here we estimate the communication requirements of implementing such DLS programs. We assume that every 1000 residential households are connected via a multiple-access channel (MAC) to  a collector node, who talks to the aggregator. Each house has 2 electric vehicles and 3 TCLs. Given the frequency of request arrivals for these appliance categories, the maximum rate requirement of the MAC channel in \eqref{ratemac} is 4.9 packets per second. This is a very low rate for a ZigBee collector, which can accommodate up to $1/0.007 \approx 142$ packets per second. What is most likely going to limit the number of houses  covered by one collector is the {\it range} at which the collector can  communicate, which is 100 meters for ZigBee. Inside this diameter, a ZigBee collector can cover 0.0142 packets per second per square meter, which can cover a 284 story building with 3 TCLs and 2 EVs in each residence. This limit will obviously not be reached in realistic scenarios. The collectors forward the $a^q(t)$'s to the aggregator every one minute through the internet. This includes 15 values for the PHEVs and 60 for the TCLs (equal to the number of clusters), amounting to less than a byte of information. The same negligible rate requirement holds for the downlink broadcast channel.

\section{Future Work}\label{sec.conclusion}
We presented a stochastic hybrid model for large populations of heterogeneous appliances, and showcased the benefits of having such a model when making market decisions and in scheduling. Future work will explore: 1) the optimal clustering method (number of clusters and bundling rule) given consumption statistics; 2) the impact of different clustering levels on market decisions and scheduling.

\appendices
\section{Proof of Lemma \ref{nxlemma}}\label{prooflemma}
 \begin{proof} Note that $x_i(t+1)=x_i(t)+\partial x_i(t)$ and $a_i(t+1)=a_i(t)+\partial a_i(t)$. Using \eqref{numberineachstate}, we can write:
\begin{align*}
\sum_{x=0}^{E}\sum_{x'=x}^{E} &\partial n_{x'}(t) = \\&~\sum_{x=0}^{E} \sum_{x'=x}^{E}\sum_{i\in {\cal P}} \delta(x_i(t+1)-x')a_i(t+1)\\
		      &-\sum_{x=0}^{E} \sum_{x'=x}^{E} \sum_{i\in {\cal P}}\delta(x_i(t)-x')a_i(t)\\
		      &=\sum_{i\in {\cal P}}\sum_{x=0}^{E}u(x_i(t)+\partial x_i(t)-x)(a_i(t)+\partial a_i(t))\\
		      &-\sum_{i\in {\cal P}}\sum_{x=0}^{E}u(x_i(t)-x)a_i(t)\\
		      &=\sum_{i\in {\cal P}}\sum_{x=0}^{E} [u(x_i(t)+\partial x_i(t)-x)-u(x_i(t)-x)]a_i(t)\\
		      &+\sum_{i\in {\cal P}}\sum_{x=0}^{E}u(x_i(t+1)  -x)\partial a_i(t)\\
		      &= \sum_{i\in {\cal P}}\partial x_i(t) a_i(t)+ \sum_{i\in {\cal P}} (x_i(t+1)+1) \partial a_i(t) \\
		      &=\sum_{i\in {\cal P}}L_i(t)+ \sum_{i\in {\cal P}} (S_i+1) \partial a_i(t)  \\
&=L(t)+ \sum_{x=0}^{E} (x+1) \partial a_x(t).
\end{align*}
The second to last equality holds following \eqref{pop} and the fact that $\partial a_i(t)$ is only non zero when $t_i = t+1$, i.e., $x_i(t+1) = S_i$.
\end{proof}

\bibliographystyle{IEEEtran}
\small
\bibliography{New2,new4,new5,newrefs}

\begin{thebibliography}{10}
\providecommand{\url}[1]{#1}
\csname url@samestyle\endcsname
\providecommand{\newblock}{\relax}
\providecommand{\bibinfo}[2]{#2}
\providecommand{\BIBentrySTDinterwordspacing}{\spaceskip=0pt\relax}
\providecommand{\BIBentryALTinterwordstretchfactor}{4}
\providecommand{\BIBentryALTinterwordspacing}{\spaceskip=\fontdimen2\font plus
\BIBentryALTinterwordstretchfactor\fontdimen3\font minus
  \fontdimen4\font\relax}
\providecommand{\BIBforeignlanguage}[2]{{%
\expandafter\ifx\csname l@#1\endcsname\relax
\typeout{** WARNING: IEEEtran.bst: No hyphenation pattern has been}%
\typeout{** loaded for the language `#1'. Using the pattern for}%
\typeout{** the default language instead.}%
\else
\language=\csname l@#1\endcsname
\fi
#2}}
\providecommand{\BIBdecl}{\relax}
\BIBdecl

\bibitem{hammerstrom2007pacific}
D.~Hammerstrom, R.~Ambrosio, J.~Brous, T.~Carlon, D.~Chassin, J.~DeSteese,
  R.~Guttromson, G.~Horst, O.~J{\"a}rvegren, R.~Kajfasz \emph{et~al.},
  ``Pacific northwest gridwise testbed demonstration projects,'' \emph{Part I.
  Olympic Peninsula Project}, 2007.

\bibitem{galus}
M.~D. Galus, S.~Koch, and G.~Andersson, ``Provision of load frequency control
  by phevs, controllable loads, and a cogeneration unit,'' \emph{Industrial
  Electronics, IEEE Transactions on}, vol.~58, no.~10, pp. 4568--4582, 2011.

\bibitem{6084772}
E.~Sortomme and M.~El-Sharkawi, ``Optimal combined bidding of vehicle-to-grid
  ancillary services,'' \emph{Smart Grid, IEEE Transactions on}, vol.~3, no.~1,
  pp. 70--79, 2012.

\bibitem{6419868}
R.~Bessa and M.~Matos, ``Optimization models for ev aggregator participation in
  a manual reserve market,'' \emph{Power Systems, IEEE Transactions on},
  vol.~28, no.~3, pp. 3085--3095, 2013.

\bibitem{6507354}
W.~Yao, J.~Zhao, F.~Wen, Y.~Xue, and G.~Ledwich, ``A hierarchical decomposition
  approach for coordinated dispatch of plug-in electric vehicles,'' \emph{Power
  Systems, IEEE Transactions on}, vol.~28, no.~3, pp. 2768--2778, 2013.

\bibitem{6145671}
P.~Sanchez-Martin, G.~Sanchez, and G.~Morales-Espana, ``Direct load control
  decision model for aggregated ev charging points,'' \emph{Power Systems, IEEE
  Transactions on}, vol.~27, no.~3, pp. 1577--1584, 2012.

\bibitem{mohsenian2010autonomous}
A.~Mohsenian-Rad, V.~W. Wong, J.~Jatskevich, R.~Schober, and A.~Leon-Garcia,
  ``Autonomous demand-side management based on game-theoretic energy
  consumption scheduling for the future smart grid,'' \emph{Smart Grid, IEEE
  Transactions on}, vol.~1, no.~3, pp. 320--331, 2010.

\bibitem{kefayati2010efficient}
M.~Kefayati and C.~Caramanis, ``Efficient energy delivery management for
  phevs,'' in \emph{Smart Grid Communications (SmartGridComm), 2010 First IEEE
  International Conference on}.\hskip 1em plus 0.5em minus 0.4em\relax IEEE,
  2010, pp. 525--530.

\bibitem{kesidis}
S.~Caron and G.~Kesidis, ``Incentive-based energy consumption scheduling
  algorithms for the smart grid,'' in \emph{Smart Grid Communications
  (SmartGridComm), 2010 First IEEE International Conference on}, 2010, pp.
  391--396.

\bibitem{chen2012large}
S.~Chen, Y.~Ji, and L.~Tong, ``Large scale charging of electric vehicles,'' in
  \emph{Power and Energy Society General Meeting, 2012 IEEE}.\hskip 1em plus
  0.5em minus 0.4em\relax IEEE, 2012, pp. 1--9.

\bibitem{subramanian2012real}
A.~Subramanian, M.~Garcia, A.~Dominguez-Garcia, D.~Callaway, K.~Poolla, and
  P.~Varaiya, ``Real-time scheduling of deferrable electric loads,'' in
  \emph{American Control Conference (ACC), 2012}.\hskip 1em plus 0.5em minus
  0.4em\relax IEEE, 2012, pp. 3643--3650.

\bibitem{ramraja}
G.~O'Brien and R.~Rajagopal, ``A method for automatically scheduling notified
  deferrable loads,'' in \emph{American Control Conference (ACC), 2013}, 2013,
  pp. 5080--5085.

\bibitem{6471273}
G.~Koutitas and L.~Tassiulas, ``Periodic flexible demand: Optimization and
  phase management in the smart grid,'' \emph{Smart Grid, IEEE Transactions
  on}, vol.~4, no.~3, pp. 1305--1313, 2013.

\bibitem{homer}
T.~Lambert, P.~Gilman, and P.~Lilienthal, \emph{Micropower System Modeling with
  Homer}.\hskip 1em plus 0.5em minus 0.4em\relax John Wiley and Sons, Inc.,
  2006, pp. 379--418.

\bibitem{6426102}
A.~Subramanian, J.~Taylor, E.~Bitar, D.~Callaway, K.~Poolla, and P.~Varaiya,
  ``Optimal power and reserve capacity procurement policies with deferrable
  loads,'' in \emph{Decision and Control (CDC), 2012 IEEE 51st Annual
  Conference on}, 2012, pp. 450--456.

\bibitem{ortega}
M.~Ortega-Vazquez, F.~Bouffard, and V.~Silva, ``Electric vehicle
  aggregator/system operator coordination for charging scheduling and services
  procurement,'' \emph{Power Systems, IEEE Transactions on}, vol.~28, no.~2,
  pp. 1806--1815, 2013.

\bibitem{Chong}
C.-Y. Chong and R.~P. Malhami, ``Statistical synthesis of physically based load
  models with applications to cold load pickup,'' \emph{Power Apparatus and
  Systems, IEEE Transactions on}, vol. PAS-103, no.~7, pp. 1621--1628, July
  1984.

\bibitem{lu-chassin04}
N.~Lu and D.~P. Chassin, ``A state-queueing model of thermostatically
  controlled appliances,'' \emph{Power Systems, IEEE Transactions on}, vol.~19,
  no.~3, pp. 1666--1673, 2004.

\bibitem{lu-chassin-windegren05}
N.~Lu, D.~P. Chassin, and S.~E. Widergren, ``Modeling uncertainties in
  aggregated thermostatically controlled loads using a state queueing model,''
  \emph{Power Systems, IEEE Transactions on}, vol.~20, no.~2, pp. 725--733,
  2005.

\bibitem{koch2011modeling}
S.~Koch, J.~L. Mathieu, and D.~S. Callaway, ``Modeling and control of
  aggregated heterogeneous thermostatically controlled loads for ancillary
  services,'' in \emph{Proc. PSCC}, 2011, pp. 1--7.

\bibitem{hao2013generalized}
H.~Hao, B.~M. Sanandaji, K.~Poolla, and T.~L. Vincent, ``A generalized battery
  model of a collection of thermostatically controlled loads for providing
  ancillary service,'' in \emph{proceedings of the 51-th Annual Allerton
  Conference on Communication, Control and Computing}, 2013.

\bibitem{globalSIP}
M.~Alizadeh and A.~Scaglione, ``Least laxity first scheduling of
  thermostatically controlled loads for regulation services,'' in \emph{IEEE
  Global Conference on Signal and Information Processing (GlobalSIP)}, 2013.

\bibitem{6407491}
J.~Foster and M.~Caramanis, ``Optimal power market participation of plug-in
  electric vehicles pooled by distribution feeder,'' \emph{Power Systems, IEEE
  Transactions on}, vol.~28, no.~3, pp. 2065--2076, 2013.

\bibitem{smartgridcomm}
M.~Alizadeh, A.~Scaglione, R.~J. Thomas, and D.~Callaway, ``Information
  infrastructure for cellular load management in green power delivery
  systems,'' in \emph{Smart Grid Communications (SmartGridComm), 2011 IEEE
  International Conference on}.\hskip 1em plus 0.5em minus 0.4em\relax IEEE,
  2011, pp. 13--18.

\bibitem{lalitha}
L.~Sankar, S.~R. Rajagopalan, S.~Mohajer, and H.~V. Poor, ``Smart meter
  privacy: A theoretical framework,'' \emph{Smart Grid, IEEE Transactions on},
  vol.~4, no.~2, pp. 837--846, 2013.

\bibitem{5622047}
G.~Kalogridis, C.~Efthymiou, S.~Denic, T.~Lewis, and R.~Cepeda, ``Privacy for
  smart meters: Towards undetectable appliance load signatures,'' in
  \emph{Smart Grid Communications (SmartGridComm), 2010 First IEEE
  International Conference on}, 2010, pp. 232--237.

\bibitem{Acs2011IDD2042445.2042457}
\BIBentryALTinterwordspacing
G.~Acs and C.~Castelluccia, ``I have a dream!: differentially private smart
  metering,'' in \emph{Proceedings of the 13th international conference on
  Information hiding}, ser. IH'11.\hskip 1em plus 0.5em minus 0.4em\relax
  Berlin, Heidelberg: Springer-Verlag, 2011, pp. 118--132. [Online]. Available:
  \url{http://dl.acm.org/citation.cfm?id=2042445.2042457}
\BIBentrySTDinterwordspacing

\bibitem{Garcia2010PEV2050149.2050164}
\BIBentryALTinterwordspacing
F.~D. Garcia and B.~Jacobs, ``Privacy-friendly energy-metering via homomorphic
  encryption,'' in \emph{Proceedings of the 6th international conference on
  Security and trust management}, ser. STM'10.\hskip 1em plus 0.5em minus
  0.4em\relax Berlin, Heidelberg: Springer-Verlag, 2011, pp. 226--238.
  [Online]. Available: \url{http://dl.acm.org/citation.cfm?id=2050149.2050164}
\BIBentrySTDinterwordspacing

\bibitem{5622064}
F.~Li, B.~Luo, and P.~Liu, ``Secure information aggregation for smart grids
  using homomorphic encryption,'' in \emph{Smart Grid Communications
  (SmartGridComm), 2010 First IEEE International Conference on}, 2010, pp.
  327--332.

\bibitem{Kursawe2011PAS2032162.2032172}
\BIBentryALTinterwordspacing
K.~Kursawe, G.~Danezis, and M.~Kohlweiss, ``Privacy-friendly aggregation for
  the smart-grid,'' in \emph{Proceedings of the 11th international conference
  on Privacy enhancing technologies}, ser. PETS'11.\hskip 1em plus 0.5em minus
  0.4em\relax Berlin, Heidelberg: Springer-Verlag, 2011, pp. 175--191.
  [Online]. Available: \url{http://dl.acm.org/citation.cfm?id=2032162.2032172}
\BIBentrySTDinterwordspacing

\bibitem{rottondi2013privacy}
C.~Rottondi and G.~Verticale, ``Privacy-friendly appliance load scheduling in
  smart grids,'' \emph{IEEE Smartgridcomm}, 2013.

\bibitem{alur1993hybrid}
R.~Alur, C.~Courcoubetis, T.~A. Henzinger, and P.-H. Ho, \emph{Hybrid automata:
  An algorithmic approach to the specification and verification of hybrid
  systems}.\hskip 1em plus 0.5em minus 0.4em\relax Springer, 1993.

\bibitem{chong85}
R.~Malhame and C.-Y. Chong, ``Electric load model synthesis by diffusion
  approximation of a high-order hybrid-state stochastic system,''
  \emph{Automatic Control, IEEE Transactions on}, vol.~30, no.~9, pp. 854--860,
  1985.

\bibitem{kleywegt2002sample}
A.~J. Kleywegt, A.~Shapiro, and T.~Homem-de Mello, ``The sample average
  approximation method for stochastic discrete optimization,'' \emph{SIAM
  Journal on Optimization}, vol.~12, no.~2, pp. 479--502, 2002.

\bibitem{alizadeh2013ev}
M.~Alizadeh, A.~Scaglione, J.~Davies, and K.~S. Kurani, ``A scalable stochastic
  model for the electricity demand of electric and plug-in hybrid vehicles,''
  \emph{Smart Grid, IEEE Transactions on}, 2013.

\bibitem{ihara1981physically}
S.~Ihara and F.~C. Schweppe, ``Physically based modeling of cold load pickup,''
  \emph{Power Apparatus and Systems, IEEE Transactions on}, no.~9, pp.
  4142--4150, 1981.

\bibitem{6149125}
K.~Kalsi, M.~Elizondo, J.~Fuller, S.~Lu, and D.~Chassin, ``Development and
  validation of aggregated models for thermostatic controlled loads with demand
  response,'' in \emph{System Science (HICSS), 2012 45th Hawaii International
  Conference on}, Jan 2012, pp. 1959--1966.

\bibitem{mathieu}
J.~Mathieu, S.~Koch, and D.~Callaway, ``State estimation and control of
  electric loads to manage real-time energy imbalance,'' \emph{Power Systems,
  IEEE Transactions on}, vol.~28, no.~1, pp. 430--440, 2013.

\bibitem{alizadeh}
M.~Alizadeh, Y.~Xiao, A.~Scaglione, and M.~van~der Schaar, ``Incentive design
  for direct load control programs,'' in \emph{Communication, Control, and
  Computing (Allerton), 2013 51th Annual Allerton Conference on}.

\bibitem{bitar2012deadline}
E.~Bitar and S.~Low, ``Deadline differentiated pricing of deferrable electric
  power service,'' in \emph{Decision and Control (CDC), 2012 IEEE 51st Annual
  Conference on}.\hskip 1em plus 0.5em minus 0.4em\relax IEEE, 2012, pp.
  4991--4997.

\bibitem{kefayati}
M.~Kefayati and R.~Baldick, ``Energy delivery transaction pricing for flexible
  electrical loads,'' in \emph{Smart Grid Communications (SmartGridComm), 2011
  IEEE International Conference on}.\hskip 1em plus 0.5em minus 0.4em\relax
  IEEE, 2011, pp. 363--368.

\bibitem{nlm}
G.~W. Hart, ``Nonintrusive appliance load monitoring,'' \emph{Proceedings of
  the IEEE}, vol.~80, no.~12, pp. 1870--1891, 1992.

\bibitem{Dingledine2004TSO1251375.1251396}
\BIBentryALTinterwordspacing
R.~Dingledine, N.~Mathewson, and P.~Syverson, ``Tor: the second-generation
  onion router,'' in \emph{Proceedings of the 13th conference on USENIX
  Security Symposium - Volume 13}, ser. SSYM'04.\hskip 1em plus 0.5em minus
  0.4em\relax Berkeley, CA, USA: USENIX Association, 2004, pp. 21--21.
  [Online]. Available: \url{http://dl.acm.org/citation.cfm?id=1251375.1251396}
\BIBentrySTDinterwordspacing

\bibitem{diff2}
C.~Dwork, ``Differential privacy: A survey of results,'' in \emph{Theory and
  Applications of Models of Computation}.\hskip 1em plus 0.5em minus
  0.4em\relax Springer, 2008, pp. 1--19.

\bibitem{data}
``Uk national grid metered half-hourly electricity demands,''
  \url{http://www.nationalgrid.com/uk/Electricity/Data/Demand+Data/}.

\end{thebibliography}

\end{document}